\documentclass{ifacconf}

\usepackage{graphicx}      
\usepackage{natbib}        
\usepackage{amsmath,amssymb,amsfonts}
\usepackage{array}
\usepackage[usenames,dvipsnames]{xcolor}
\usepackage[shortlabels]{enumitem}
\usepackage{tikz}
\usepackage{pgfplots}
\usepackage{url}
\usepackage{varwidth}

\pgfplotsset{compat=1.16}
\usetikzlibrary{decorations.pathreplacing,calligraphy,arrows,calc}
\def\R{\ensuremath{\mathbb{R}}}
\def\N{\ensuremath{\mathbb{N}}}
\def\X{\ensuremath{\mathbb{X}}}

\def\A{\ensuremath{\mathcal{A}}}
\def\taumiet{\ensuremath{\tau_\text{MIET}^i}}

\def\endstatement{\hfill$\Box$}
\def\yh{\ensuremath{\kern1.5pt\widehat{\kern-0.8pt y}\kern0.8pt}}
\def\what{\ensuremath{\widehat{\kern-0.5pt w}}}
\def\Nc{\ensuremath{\mathcal{N}}}

\let\leq\leqslant
\let\geq\geqslant

\newcommand{\yhs}[1]{\yh^{\kern-0.3pt#1}}
\newcommand{\yhss}[2]{\yh^{\kern-0.3pt#1}_{\kern-1.4pt#2}}
\newcommand{\yhsub}[1]{\yh_{\kern-1.4pt#1}}

\DeclareMathOperator{\diag}{diag}

\DeclareMathOperator*{\argmax}{arg~max}

\newtheorem{assumption}{Assumption}
\newtheorem{condition}{Condition}
\newtheorem{lemma}{Lemma}
\newtheorem{problem}{Problem}
\newtheorem{definition}{Definition}
\newtheorem{remark}{Remark}

\newenvironment{proof}{\textit{Proof.}}{\qed}

\usepackage[normalem]{ulem}

\begin{document}
\begin{frontmatter}

\title{Distributed Periodic Event-triggered Control of Nonlinear Multi-Agent Systems\thanksref{footnoteinfo}}

\thanks[footnoteinfo]{This work is supported by the ANR grant HANDY 18-CE40-0010.}

\author[TUe]{Koen J.A. Scheres}
\author[TUe]{Victor S. Dolk}
\author[TUe]{Michelle S. Chong}
\author[CRAN]{Romain Postoyan}
\author[TUe]{W.P. Maurice H. Heemels}

\address[TUe]{Department of Mechanical Engineering, Eindhoven University of Technology, The Netherlands. (k.j.a.scheres@tue.nl).}
\address[CRAN]{Universit\'e de Lorraine, CNRS, CRAN, F-54000 Nancy, France.}

\begin{abstract}                
We present a general emulation-based framework to address the distributed control of multi-agent systems over packet-based networks. We consider the setup where information is only transmitted at (non-uniform) sampling times and where packets are received with unknown delays. We design local dynamic periodic event-triggering mechanisms to generate the transmissions. The triggering mechanisms can run on non-synchronized digital platforms, i.e., we ensure that the conditions must only be verified at asynchronous sampling times, which may differ for each platform. Different stability and performance characteristics can be considered as we follow a general dissipativity-based approach. Moreover, Zeno-free properties are guaranteed by design. The results are illustrated on a consensus problem.
\end{abstract}


\end{frontmatter}

\section{Introduction}
Distributed and multi-agent control systems, including the consensus problem, have attracted a lot of attention in recent years. When these systems communicate via a packet-based network, information is sparsely available. In these cases, event-triggered control can be used. Event-triggered control consists of generating the input signal and updating it over the packet-based network at some time instants, which are based on the available plant/agent information, to guarantee relevant stability and performance properties, see, e.g., \cite{Heemels_Johansson_Tabuada_2012}. A triggering condition is thus synthesized and monitored to decide whether a new transmission is needed. Due to the fact that these conditions are often processed on a digital platform, it is essential to take the sampling behavior of the platform into account, especially when fast sampling is not possible, e.g. in case of low-power electronics, in which case we talk of periodic event-triggered control, see, e.g., \cite{Heemels_Donkers_Teel_2013}. Moreover, in practice, the communication network suffers from imperfections such as time-varying and unknown delays, which may destroy stability of the closed-loop system. While there is abundant literature on the event-triggered control of multi-agent systems, see, e.g., \cite{Nowzari_Garcia_Cortes_2019} and references therein, to the best of our knowledge, none of the proposed approaches in the literature consider \emph{all} of the following aspects:
\begin{enumerate}
    \item nonlinear multi-system setup,
    \item fully distributed and asynchronous configuration,
    \item implementability on digital platforms,
    \item unknown and time-varying sampling and transmission delays,
    \item general stability and performance properties for (possibly unbounded) attractors (as in consensus).
\end{enumerate}
Many works treat a subset of these aspects. A prominent example is, e.g., the recent work by \cite{Yu_Chen_2021}, which provides important advancements on the nonlinear case with (large) unknown transmission delays while taking sampling into account. The main focus of our paper is providing a unifying framework addressing all of these aspects.

The main contribution in this paper is the development of a unified framework for the design of Zeno-free, decentralized and asynchronous periodic event-triggering mechanisms that can be implemented on local digital platforms. The setup proposed in this paper captures a wide range of relevant multi-agent (but also centralized) control problems by adopting a general dissipativity-based framework. Using this framework, we can consider distributed stabilization of nonlinear systems, output-regulation problems (of which the consensus-seeking problem is a particular case) and vehicle-platooning problems (in which $\mathcal{L}_p$-contractivity, $p\in[1,\infty)$, is of interest as a string stability guarantee).
A notable advantage of our setup is that clock synchronization is not required. Hence each local platform can sample and transmit independently of all other platforms, making the algorithm fully distributed.
\section{Preliminaries}
\subsection{Notation}
The sets of all non-negative and positive integers are denoted $\N$ and $\N_{>0}$, respectively. The fields of all reals, all non-negative reals and all non-positive reals are indicated by $\R$, $\R_{\geq0}$ and $\R_{\leq0}$, respectively. The identity matrix of size $N\times N$ is denoted by $I_N$, and the vectors in $\R^N$ whose elements are all ones or zeros are denoted by $\mathbf{1}_N$ and $\mathbf{0}_N$, respectively. For $N$ vectors $x_i\in\R^{n_i}$, we use the notation $(x_1,x_2,\ldots,x_N)$ to denote $\begin{bmatrix}x_1^\top&x_2^\top&\ldots&x_N^\top\end{bmatrix}^\top$. Given matrices $A_1,\ldots,A_n$, we denote by $\diag(A_1,\ldots,A_n)$ the block-diagonal matrix where the main diagonal blocks consist of the matrices $A_1$ to $A_n$ and all other blocks are zero matrices. By $\langle\cdot,\cdot\rangle$ and $|\cdot|$ we denote the usual inner product of real vectors and the Euclidean norm, respectively. We denote the logical \emph{and} and \emph{or} operators as $\land$ and $\lor$, respectively. For two matrices $A\in\R^{m\times n}$ and $B\in\R^{p\times q}$, the Kronecker product of $A$ with $B$ is denoted $A\otimes B\in\R^{mp\times nq}$. The cardinality of a finite set $\mathcal{S}$ is denoted $\vert\mathcal{S}\vert$. The notation $F:X\rightrightarrows Y$, indicates that $F$ is a set-valued mapping from $X$ to $Y$ with $F(x)\subseteq Y$ for all $x\in X$. For any $x\in\R^n$, the distance to a closed non-empty set $\A$ is denoted by $|x|_\A:=\min_{y\in\A}|x-y|$. We use $U^{\circ}(x;v)$ to denote the generalized directional derivative of Clarke of a locally Lipschitz function $U$ at $x$ in the direction $v$, i.e., $U^{\circ}(x;v):=\lim\sup_{h\rightarrow 0^+,\, y\rightarrow x}(U(y+hv)-U(y))/h$, which reduces to the standard directional derivative $\left\langle \nabla U(x),v\right\rangle$ when $U$ is continuously differentiable.

\subsection{Graph Theory}
A graph $\mathcal{G}:=(\mathcal{V},\mathcal{E})$ consists of a vertex set $\mathcal{V}:=\{1,2,...,N\}$ with $N\in\N_{>0}$ and a set of edges $\mathcal{E}\subset\mathcal{V}\times\mathcal{V}$. An ordered pair $(i,j)\in\mathcal{E}$ with $i,j\in\mathcal{V}$ is said to be an edge \emph{directed} from $i$ to $j$. A graph is called \emph{undirected} if it holds that $(i,j)\in\mathcal{E}$ if and only if $(j,i)\in\mathcal{E}$. Otherwise, the graph is a \emph{directed} graph, also referred to as a digraph. A sequence of edges $(i,j)\in\mathcal{E}$ connecting two vertices is called a directed path. For a connected graph $\mathcal{G}$, there exists a path between any two vertices in $\mathcal{V}$. A vertex $j$ is said to be an \emph{out}-neighbor of $i$ if $(i,j)\in\mathcal{E}$ and an \emph{in}-neighbor of $i$ if $(j,i)\in\mathcal{E}$. The set $\mathcal{V}^\text{in}_i$ of the in-neighbors of $i$ is defined as $\mathcal{V}^\text{in}_i:=\{j\in\mathcal{V}\;|\;(j,i)\in\mathcal{E}\}$, and the set $\mathcal{V}^\text{out}_i$ of out-neighbors of $i$ is defined as $\mathcal{V}^\text{out}_i:=\{j\in\mathcal{V}\;|\;(i,j)\in\mathcal{E}\}$. The cardinality of $\mathcal{V}_i^\text{out}$ is denoted as $N_i$.

\subsection{Hybrid systems}
We consider hybrid systems $\mathcal{H}(\mathcal{C},F,\mathcal{D},G)$ given by
\begin{subequations}\label{eq:HybridFramework}
	\begin{align}
	\dot{\xi}&= F(\xi,v), &&\text{when } \xi\in \mathcal{C},\label{eq:Flow}\\
	\xi^+&\in G(\xi), &&\text{when }\xi\in \mathcal{D},\label{eq:Jump}
	\end{align}
\end{subequations}
where $F$ and $G$ denote the flow and the jump map, respectively, $\mathcal{C}\subseteq\X$ and $\mathcal{D}\subseteq\X$ the flow and the jump set, respectively, see \cite{Goebel_Sanfelice_Teel_2012}. We adopt the notion of solutions recently proposed in \cite{Heemels_Bernard_Scheres_Postoyan_Sanfelice_2021} for hybrid systems with inputs. For these hybrid systems, we are interested in the following dissipativity property, which is close in nature to the one used in \cite{Teel_2010}.
\begin{definition}\label{defn:Vdiss}
    Let $s:\X\times\R^{n_v}\to\R$ be a \emph{supply rate} and $\mathcal{S}\subseteq\X$ be a closed non-empty set. System $\mathcal{H}$ is said to be \emph{$s$-flow-dissipative} with respect to $\mathcal{S}$, or in short, ($s,\mathcal{S}$)-flow-dissipative, if there exists a locally Lipschitz function $U:\mathbb{X}\to \R_{\geqslant0}$, called a \emph{storage function}, such that
        \begin{enumerate}
            \item there exist $\mathcal{K}_\infty$-functions $\underline{\alpha}$, $\overline{\alpha}$ such that for all $\xi\in\mathbb{X}$,
            $\underline{\alpha}(|\xi|_\mathcal{S}) \leq U(\xi) \leq \overline{\alpha}(|\xi|_\mathcal{S})$, where $|\xi|_\mathcal{S}$ denotes the distance of $\xi$ to the set $\mathcal{S}$,
            \item for all $\xi\in\mathcal{C}$, $v\in\R^{n_v}$ and $f\in F(\xi,v)$,
            $U^{\circ}(\xi;f) \leqslant s(\xi,v)$, where $U^\circ$ denotes the generalized directional derivative of Clarke,
            \item for all $\xi\in\mathcal{D}$ and all $g\in G(\xi)$,
                $U(g) - U(\xi) \leqslant 0$.\endstatement
        \end{enumerate}
\end{definition}
\section{System setup}\label{sec:MASsetup}
\subsection{Multi-agent systems}
We consider the setting where multiple agents, each with a local digital platform, communicate with each other via a packet-based network to achieve a common goal such as stabilization, consensus, $\mathcal{L}_p$-performance, etc., which will be captured by a dissipativity property as in Definition \ref{defn:Vdiss}, as explained further below. To be precise, we consider a collection of $N\in\N_{>0}$ heterogeneous agents $\A_1,\A_2,\ldots,\A_N$, which are interconnected according to a digraph $\mathcal{G}(\mathcal{V},\mathcal{E})$ where $\mathcal{V}:=\{1,2,\ldots,N\}$. The digital platform of each agent is used for the implementation of algorithms and control computations. Due to the digital nature, measurements are not continuously available, but only on specific sampling times, as explained in more detail in Section \ref{sec:digitalPlatform}. The dynamics of the $i^\text{th}$ agent, $i\in\Nc:=\{1,2,\ldots,N\}$, are given by
\begin{equation}\label{eq:agenti}
    \A_i:\begin{cases}
        \begin{aligned}
            \dot{x}_i&=f_i(x,\yhss{\text{in}}{i},v_i),\\
            y_i&=h_i(x_i),
        \end{aligned}
    \end{cases}
\end{equation}
where $x_i\in\R^{n_{x,i}}$ is the local state vector, $x:=(x_1,x_2,\ldots,x_N)\in\R^{n_x}$ with $n_x:=\sum_{i\in\Nc}n_{x,i}$ is the global state vector, $v_i\in\R^{n_{v,i}}$ is a local exogenous disturbance or input, $y_i\in\R^{n_{y,i}}$ is the local output, $y:=(y_1,y_2,\ldots,y_N)\in\R^{n_y}$ with $n_y:=\sum_{i\in\Nc}n_{y,i}$ is the global output and $\yhss{\text{in}}{i}\in\R^{n_y}$ is agent $\A_i$'s estimate of the outputs of agents $\A_m$, $m\in\mathcal{V}^\text{in}_i$, where $\mathcal{V}_i^\text{in}$ denotes the collection of all agents that transmit information to agent $\A_i$, as will be explained in more detail in Section \ref{sec:outputs}. We assume that the functions $f_i$ are continuous and that the functions $h_i$ are continuously differentiable. The maps $f_i$ may depend on the entire vector $x$, implying that we can allow physical couplings between agents, for example in an interconnected physical plant, see Fig. \ref{fig:setup} for an illustration.
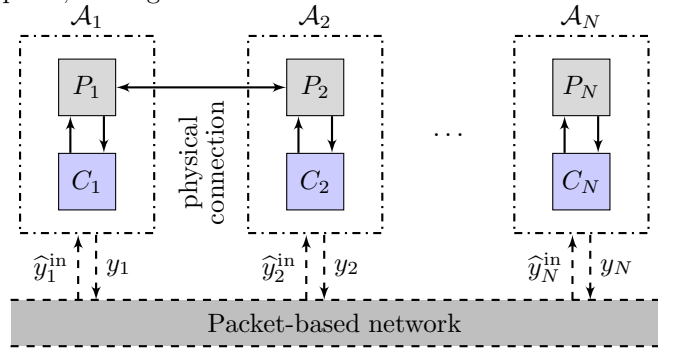
\begin{figure}[ht!]
    \centering
%
%
\vspace{-3mm}
\tikzstyle{controller} = [draw, fill=blue!20, rectangle, minimum height=0.75cm, minimum width=0.75cm]
\tikzstyle{dots} = [draw=none, rectangle, minimum height=0.75cm, minimum width=0.75cm]
\tikzstyle{plant} = [draw, fill=gray!30, rectangle, minimum height=0.75cm, minimum width=0.75cm]
\tikzstyle{network} = [draw=none, fill=gray!50, rectangle, minimum height=0.6cm, minimum width=8.5cm]
\tikzstyle{generator} = [draw, rectangle, minimum height=0.5cm, minimum width=1cm]
\begin{tikzpicture}[node distance=1.5cm,>=latex']
    \node[plant] (p1) {$P_1$};
    \node[controller, below of=p1, node distance=1.25cm] (c1) {$C_1$};
    \draw[thick, dashdotted] ($(p1.north west)+(-0.5,0.3)$) rectangle ($(c1.south east)+(0.5,-0.3)$) node[midway,align=center,label={[yshift=1.15cm]$\mathcal{A}_1$}] (a1) {};
    \draw [->, thick] (p1.300)--(c1.60);
    \draw [<-, thick] (p1.240)--(c1.120);

    \node[plant, right of=p1, node distance=3cm] (p2) {$P_2$};
    \node [controller, below of=p2, node distance=1.25cm] (c2) {$C_2$};
    \draw[thick, dashdotted] ($(p2.north west)+(-0.5,0.3)$) rectangle ($(c2.south east)+(0.5,-0.3)$) node[midway,align=center,label={[yshift=1.15cm]$\mathcal{A}_2$}] (a2) {};
    \draw [->, thick] (p2.300)--(c2.60);
    \draw [<-, thick] (p2.240)--(c2.120);

    \draw[<->, thick] (p1)--(p2) node[align=center, midway, rotate=90, xshift=-1cm, yshift=0em] {physical\\connection};

    \node[right of=a2, node distance=1.75cm] {$\ldots$};

    \node[plant, right of=p2, node distance=3.5cm] (pn) {$P_N$};
    \node[controller, below of=pn, node distance=1.25cm] (cn) {$C_N$};
    \draw[->, thick] (pn.300)--(cn.60);
    \draw[<-, thick] (pn.240)--(cn.120);
    \draw[thick, dashdotted] ($(pn.north west)+(-0.5,0.3)$) rectangle ($(cn.south east)+(0.5,-0.3)$) node[midway,align=center,label={[yshift=1.15cm]$\mathcal{A}_N$}] (an) {};

    \node[network, yshift=-2.5cm] at ($(a1)!0.5!(an)$) (net) {Packet-based network};

    \node at ($(c1.south west)+(0.5,-0.2)$) (y1) {};
    \node at ($(c1.south east)+(-0.5,-0.2)$) (yin1) {};
    \draw[->, thick, dashed] (y1)--(net.north-|y1) node[midway, right] {$y_1$};
    \draw[->, thick, dashed] (net.north-|yin1)--(yin1) node[midway, left] {$\yhss{\mathrm{in}}{1}$};

    \node at ($(c2.south west)+(0.5,-0.2)$) (y2) {};
    \node at ($(c2.south east)+(-0.5,-0.2)$) (yin2) {};
    \draw[->, thick, dashed] (y2)--(net.north-|y2) node[midway, right] {$y_2$};
    \draw[->, thick, dashed] (net.north-|yin2)--(yin2) node[midway, left] {$\yhss{\mathrm{in}}{2}$};

    \node at ($(cn.south west)+(0.5,-0.2)$) (yn) {};
    \node at ($(cn.south east)+(-0.5,-0.2)$) (yinn) {};
    \draw[->, thick, dashed] (yn)--(net.north-|yn) node[midway, right] {$y_N$};
    \draw[->, thick, dashed] (net.north-|yinn)--(yinn) node[midway, left] {$\yhss{\mathrm{in}}{N}$};

    \draw[thick, dashed] (net.north west)--(net.north east);
    \draw[thick, dashed] (net.south west)--(net.south east);

\end{tikzpicture}
\vspace{-2mm}
    \caption{\small Example of a networked control setup with several agents connected to a packet-based network and a physical connection between $\A_1$ and $\A_2$.}
    \label{fig:setup}
\end{figure}

Typical systems that can be captured by a multi-agent system are, e.g., (physically uncoupled) agents, a fleet of vehicles seeking consensus, or a distributed plant (with physical couplings) with distributed controllers. To design the controllers, we proceed with an emulation-based approach in which a (local) controller is designed such that, in the absence of a network, i.e., when $\yhss{\text{in}}{i}=y$ for all $i\in\mathcal{V}_i^\text{in}$, the system satisfies the desired stability and performance criteria. Since the controller is implemented on a digital platform, we assume that either the controller is static and updated during the sampling times of the output (see Section \ref{sec:digitalPlatform}), or, when the (local) controllers are dynamic, that they run on a separate platform, whose sampling times are much faster and hence they operate in (near) continuous-time.

\subsection{Transmitted outputs}\label{sec:outputs}
Due to the usage of a packet-based network, agent $\A_i$, $i\in\Nc$, does not have continuous access to the output $y_m$ of agent $\A_m$, $m\in\mathcal{V}^\text{in}_i$. Instead, agent $\A_i$ only has an estimate $\yhss{i}{m}$ of $y_m$, $m\in\mathcal{V}^\text{in}_i$, collected in the estimation vector $\yhss{\text{in}}{i}:=(\yhss{i}{1},\yhss{i}{2},\ldots,\yhss{i}{N})\in\R^{n_y}$. For all $m\in\mathcal{N}\setminus\mathcal{V}^\text{in}_i$, $\yhss{i}{m}$ is not relevant and simply set to zero.

At discrete times $t_k^i$, $k\in\N$, $i\in\Nc$, that satisfy $0=t_0^i<t_1^i<\ldots$, the output $y_i$ of agent $\mathcal{A}_i$ is broadcasted over the network to all (connected) agents $\A_m$, where $m\in\mathcal{V}^\text{out}_i$ with $\mathcal{V}_i^\text{out}$ the collection of all agents that receive information from agent $\A_i$. Due to possible network delays, the estimate $\yhss{m}{i}$, $m\in\mathcal{V}^\text{out}_i$, is updated after a delay of $\Delta^{i,m}_k\geq0$.  Note that the delays are different for each (receiving) agent. The update of the estimate $\yhss{m}{i}$, $i\in\Nc$, $m\in\mathcal{V}^\text{out}_i$, can be expressed as
\begin{equation}
    \yhss{m}{i}((t_k^i+\Delta^{i,m}_k)^+)=y_i(t_k^i).
\end{equation}
In between transmissions, the estimate $\yhss{m}{i}$ is generated by a zero-order-hold (ZOH) device, i.e.,
\begin{equation}\label{eq:holdingfunction}
    \dot{\kern7pt\yhss{m}{i}}(t)=0,
\end{equation}
for all $t\in(t_k^i+\Delta^{i,m}_k,t_{k+1}^i+\Delta^{i,m}_{k+1})$, with $i\in\Nc$, $m\in\mathcal{V}^\text{out}_i$, $k\in\N$.

The transmission times $t_k^i$ are determined by an \emph{event generator} or \emph{triggering mechanism}. Inspired by \cite{Girard_2015} and \cite{Dolk_Borgers_Heemels_2017}, we consider dynamic event triggering rules, where an auxiliary variable $\eta_i\in\R_{\geq0}$, $i\in\Nc$, whose dynamics are designed in the sequel, is used to determine the transmission times $t_k^i$, $k\in\N$, see Section \ref{sec:triggering}.

\subsection{Digital platform}\label{sec:digitalPlatform}
The triggering mechanism of each agent is implemented on the local digital platform, which has its own sampling times. The sequence of sampling times of agent $\A_i$ is denoted $\{s_n^i\}_{n\in\N}$, where $s_n^i$ denotes the $n^\text{th}$ local sampling instant of agent $\A_i$. Transmissions generated by $\mathcal{A}_i$ occur on a subset of the sampling instants, i.e.,
\begin{equation}\label{eq:transmissiontimes}
    \{t_k^i\}_{k\in\N}\subseteq\{s_n^i\}_{n\in\N}.
\end{equation}
Inspired by \cite{Wang_Postoyan_Nesic_Heemels_2020}, we consider the general setting where the inter-sampling times satisfy,
\begin{equation}\label{eq:masp}
    0<d_i\leq s_{n+1}^i-s_n^i\leq \tau_\text{MASP}^i,
\end{equation}
where $d_i$ is an arbitrarily small but positive constant and $\tau_\text{MASP}^i$ denotes the \emph{maximum allowable sampling period} (MASP) for agent $\A_i$, $i\in\Nc$. The sampling times $\{s_n^i\}_{n\in\N}$ and $\{s_n^j\}_{n\in\N}$ of agents $\A_i$ and $\A_j$, respectively, are a priori not related for $i\neq j$. In other words, all agents operate \emph{independently} and \emph{asynchronously}.

Due to the agents operating asynchronously, the arrival times $t_k^i+\Delta^{i,m}_k$, $k\in\N$, of new information at agent $\A_m$ from agent $\A_i$ may not coincide with the sampling times $\{s_n^m\}_{n\in\N}$ of agent $\A_m$, hence information may be received in between consecutive sampling times of agent $\A_m$. However, the sampling-induced delay (the time between the arrival of information from agent $\A_i$ and the next sampling instant of agent $\A_m$) can be included in the total delay denoted $\overline{\Delta}^{i,m}_k$. Therefore, the total delay $\overline{\Delta}^{i,m}_k$ is equal to the combined communication delay $\Delta^{i,m}_k$ and sampling-induced delay. Through this setup, we obtain
\begin{equation}\label{eq:receivingtimes}
    \{t_k^i+\overline{\Delta}^{i,m}_k\}_{k\in\N}\subseteq\{s^m_n\}_{n\in\N}
\end{equation}
for all $m\in\Nc$ and $i\in\mathcal{V}^\text{out}_m$. 

We adopt the following assumption on the total delays $\overline{\Delta}^{i,m}_k$, $k\in\N$.
\begin{assumption}\label{ass:smalldelay}
    For each $i\in\Nc$, there is a time-constant $\tau_\text{MAD}^i$ such that the delays are bounded according to $0\leq\overline{\Delta}^{i,m}_k\leq\tau_\text{MAD}^i\leq t_{k+1}^i-t_k^i$ for all $m\in\mathcal{V}^\text{out}_i$ and all $k\in\N$, where $\tau_\text{MAD}^i$ denotes the \emph{maximum allowable delay} (MAD) for agent $\A_i$.\endstatement
\end{assumption}
Assumption \ref{ass:smalldelay} is a ``small delay'' condition, which also implies that packets sent from $\A_i$ to $\A_m$, $m\in\mathcal{V}^\text{out}_i$, are received in the same order that they are transmitted.

Since the sampling-induced delays are never larger than the local MASP $\tau_\text{MASP}^m$ at agent $m$, we have that
\begin{equation}\label{eq:madmasp}
    \tau_\text{MAD}^i\geq\tau_\text{MASP}^m+\Delta^{i,m}_k~\text{for all}~i\in\Nc,m\in\mathcal{V}^\text{out}_i,k\in\N.
\end{equation}
\subsection{Triggering rule}\label{sec:triggering}
Our goal is to employ dynamic event triggering, which relies on locally available information, namely output measurements. Due to this information only being available at the sampling instants $\{s^i_n\}_{n\in\N}$, the design of, e.g., \cite{Dolk_Borgers_Heemels_2017} cannot be directly used. Instead, we consider an event-triggering mechanism (ETM) in the form
\begin{equation}\label{eq:ETCCondition}
    \begin{aligned}
        t_{k+1}^i&:=\inf\{t\geq t_k^i+\tau_\text{MIET}^i\mid\\
        &\eta_i(t)+\nu_i(y_i(t),\yhss{\text{out}}{i}(t),\tau_i(t))\leq0,t\in\{s_n^i\}_{n\in\N}\},
    \end{aligned}
\end{equation}
for $i\in\Nc$, $k\in\N$, with $t_0^i=0$ and where $\tau_\text{MIET}^i>0$ denotes the (enforced lower bound on the) \emph{minimum inter-event time} (MIET) of agent $\A_i$, $\eta_i\in\R_{\geq0}$ is the auxiliary variable mentioned earlier, $\yhss{\text{out}}{i}:=(\yhss{1}{i},\yhss{2}{i},\ldots,\yhss{N}{i})$ is the vector of estimates of the output $y_i$ at the agents $\A_m$, $m\in\mathcal{V}^\text{out}_i$. Variable $\tau_i\in\R_{\geq0}$ is a local timer that is set to zero after each transmission of the output $y_i$ over the network, and whose dynamics are given by $\dot{\tau}_i=1$ in between two successive transmissions of agent $\A_i$. The function $\nu_i:\R^{n_{y,i}}\times\R^{Nn_{y,i}}\times\R_{\geq0}\to\R_{\leq0}$ is to be designed.

At first glance it might seem unnatural that agent $\A_i$ has to know the estimates $\yhss{\text{out}}{i}$ due to the presence of the unknown and time-varying delays. However, this information is only needed when $\tau_i\geq\tau_\text{MIET}^i$, and since $\tau_\text{MIET}^i\geq\tau_\text{MAD}^i$ as we will see in Section \ref{sec:DesignConditions}, all agents $\A_m$, $m\in\mathcal{V}^\text{out}_i$, will have received the latest transmission of agent $\A_i$ for $\tau_i\geq\tau_\text{MIET}^i$. Due to the fact that a ZOH is used to generate the estimated output, see \eqref{eq:holdingfunction}, $\yhss{\text{out}}{i}$ is thus the vector consisting of the last value of the output $y_i$ sent by agent $\A_i$, which is available at agent $\A_i$ for $\tau_i\geq\tau_\text{MIET}^i$.

The ETM \eqref{eq:ETCCondition} satisfies the constraints that arise from the usage of a digital platform, as the trigger condition in \eqref{eq:ETCCondition} only has to be evaluated at the local sampling times $s_n^i$, $n\in\N$. The triggering variable $\eta_i$ generated locally by agent $\A_i$, $i\in\Nc$, evolves according to
\begin{subequations}\label{eq:ETCTriggerDesign}
    \begin{align}
        \dot{\eta}_i&=\Psi_i(\yhss{\text{in}}{i})-\varphi_i(\eta_i),\label{eq:etadot}\\
        \eta_i(t^+)&\in\left\{\begin{array}{l}
            \{\eta_i+\varrho_i(y_i,\yhss{\text{out}}{i})\},\,\text{for all}~t\in\{t_k^i\}_{k\in\N},\\
            \{\eta_i+\nu_i(y_i,\yhss{\text{out}}{i},\tau_i)\},\\
            \hspace{13.5mm}\text{for all}~t\in\{s_n^i\}_{n\in\N}\setminus\{t_k^i\}_{k\in\N},
        \end{array}\right.
    \end{align}
\end{subequations}
where the functions $\Psi_i:\R^{n_y}\to\R$, $\varrho_i:\R^{n_{y,i}}\times\R^{Nn_{y,i}}\to\R_{\geq0}$, $\varphi_i\in\mathcal{K}_\infty$ and the constant $\tau_\text{MIET}^i\in\R_{>0}$ are designed in Section \ref{sec:ETMdesign}.
\begin{remark}
    In \eqref{eq:etadot}, a continuous-time differential equation is used. However, since the `external' variable $\yhss{\text{in}}{i}$ is constant in between consecutive sampling times, exact discretization or numerical integration can be used to compute $\eta_i$ a posteriori based on the elapsed time since the previous sampling time. For instance, if $\varphi_i(\eta_i)=\alpha_i\eta_i$ with a constant $\alpha_i\neq0$, we obtain the exact discretization $\eta_{i}(s_{n+1}^i) = e^{-\alpha_i(s_{n+1}^i- s_n^i)} \eta_{i}((s_{n}^{i})^+) +\alpha_i^{-1}[1-e^{-\alpha_i(s_{n+1}^i- s_n^i)}]\Psi_{i}(\yhss{\text{in}}{i}((s_n^i)^+))$. Hence, exact solutions to the differential equation can be obtained on a digital platform. However, we consider the dynamics of $\eta_i$ as presented in \eqref{eq:ETCTriggerDesign} to facilitate the modeling and stability analysis later on.\endstatement
\end{remark}
\subsection{Objective}\label{sec:objective}
Given the descriptions above, the problem considered in this paper can now be stated informally as follows. Consider a collection of maximum allowable delays $\tau_\text{MAD}^i$, $i\in\Nc$, satisfying Assumption \ref{ass:smalldelay}. Our objective is to propose design conditions for the time constants $\tau_\text{MIET}^i(\geqslant \tau_\text{MAD}^i)$, the functions $\Psi_i$, $\varphi_i$, $\varrho_{i}$ and $\nu_i$, $i\in\Nc$, as in \eqref{eq:ETCCondition} and \eqref{eq:ETCTriggerDesign}, such that the resulting system has the desired (and to be specified) closed-loop stability, performance and robustness properties formalized in terms of suitable dissipativity properties.

\section{Hybrid modeling}
To facilitate the modeling of the overall networked system, some helpful notation is introduced in this section.
\subsection{Network-induced errors}
For all $i\in\Nc$ and $m\in\mathcal{V}^\text{out}_i$, we denote the \emph{network-induced error} $e_i^m$ as the difference between the output $y_i$ of agent $\A_i$ and the estimate $\yhss{m}{i}$ of the output $y_i$ at agent $\A_m$. For all $i\in\Nc$ and $m\in\Nc\setminus\mathcal{V}^\text{out}_i$, i.e., for all redundant variables, we set $e^m_i=0$. Hence, we have
\begin{equation}\label{eq:networkerror}
    e^m_i:=\begin{cases}\yhss{m}{i}-y_i,&\text{if }m\in\mathcal{V}^\text{out}_i,\\0,&\text{if }m\in\Nc\setminus\mathcal{V}^\text{out}_i.\end{cases}
\end{equation}
We define two separate concatenations of the network-induced error associated to agent $\A_i$, $i\in\Nc$. The first one, denoted $e_i^\text{out}:=(e_i^1,e_i^2,\ldots,e_i^N)\in\overline{\mathbb{E}}_i$, where $\overline{\mathbb{E}}_i:=\mathbb{E}_i(1)\times\mathbb{E}_i(2)\times\ldots\times\mathbb{E}_i(N)$ and with
\begin{equation*}
   \mathbb{E}_i(m):=\begin{cases}\R^{n_{y,i}},&\text{if }m\in\mathcal{V}^\text{out}_i,\\\{\mathbf{0}_{n_{y,i}}\},&\text{otherwise},\end{cases}
\end{equation*}
is the concatenation of the network-induced errors associated to the output $y_i$. The second, denoted $e^\text{in}_i:=(e^i_1,e^i_2,\ldots,e^i_N)\in\mathbb{E}_i$, with $\mathbb{E}_i:=\mathbb{E}_1(i)\times\mathbb{E}_2(i)\times\ldots\times\mathbb{E}_N(i)$, is the concatenation of network-induced errors of the estimated outputs available at agent $\A_i$, $i\in\mathcal{N}$. Moreover, we define the concatenation of all network-induced errors $e_i^\text{in}$, for $i\in\Nc$, as $e:=(e^\text{in}_1,e^\text{in}_2,\ldots,e^\text{in}_N)\in\mathbb{E}$ with $\mathbb{E}:=\mathbb{E}_1\times\mathbb{E}_2\times\ldots\times\mathbb{E}_N$. Observe that $|e|^2=\sum_{i\in\Nc}|e^\text{out}_i|^2=\sum_{i\in\Nc}|e^\text{in}_i|^2$.

\subsection{Clock variables}

To be able to cast the overall system described in Section \ref{sec:MASsetup} in the form of $\mathcal{H}(\mathcal{C},F,\mathcal{D},G)$, we need to introduce some auxiliary variables. Each agent $\A_i$, $i\in\Nc$, has two local timers. We already saw that $\tau_i$ captures the time elapsed since the last transmission of agent $\A_i$, see \eqref{eq:ETCCondition}. The second, denoted $\sigma_i$, keeps track of the time elapsed since the last sampling instant of agent $\A_i$, i.e., $\dot{\sigma}_i(t)=1$ for all $t\in\R\setminus\{s_n^i\}_{n\in\N}$ and is reset to zero at each sampling instant, i.e., $\sigma_i(t^+)=0$ for all $t\in\{s_n^i\}_{n\in\N}$. Observe that $\tau_i$ takes values in $\R_{\geq0}$ and that $\sigma_i$ takes values in $\mathbb{T}_i:=[0,\tau_\text{MASP}^i]$ due to \eqref{eq:masp}. Their concatenations are defined as $\tau:=(\tau_{1},\tau_{2},\ldots,\tau_{N})\in\R_{\geq0}^{N}$ and $\sigma:=\mathbb{T}$ with $\mathbb{T}:=\mathbb{T}_1\times\mathbb{T}_2\times\ldots\times\mathbb{T}_N$.

\subsection{Indicator variables}
We also define two indicator variables, $\ell_i^m\in\{0,1\}$ and $b_i^m\in\{0,1\}$. The variable $\ell_i^m$ is used to indicate whether the most recently transmitted output value $y_i$ of agent $\A_i$ has been received by agent $\A_m$ ($\ell_i^m=0$), or that it still has to be received by agent $\A_m$ ($\ell_i^m=1$). Since information received by agent $\A_m$ is \emph{processed} at the sampling times, we assume that information is buffered if it is received between sampling instances. The variable $b_i^m$ indicates whether agent $\A_m$ will process (i.e., update $\yhss{m}{i}$) the most recently transmitted output value by $\A_i$ ($b_i^m=1$) or that $\yhss{m}{i}$ will not be updated at its next sampling instance ($b_i^m=0$). We distinguish between these two ``events'' to ensure that updates of $\yhs{m}$ align with the sampling times of agent $\A_m$, as described in \eqref{eq:receivingtimes}. A graphical representation of $\ell_i^m$ and $b_i^m$ is drawn in Fig. \ref{fig:ellandb}.
\begin{figure}[ht!]
    \centering
\vspace{-6mm}
\begin{tikzpicture}[scale=1,rotLabel/.style={above,rotate=45,anchor=200}]
    \draw[line width=3pt, black!60, opacity=0.7, line cap=round] (4,0)--(7,0);
    \draw[line width=3pt, black!60, opacity=0.7, line cap=round] (8,0)--(8.5,0);

    \draw[thick] (0,0) -- (8.5,0);

    \foreach \x in {0,1.16,1.7,2.3,3.2,4.1}
        \filldraw[NavyBlue] (1.8*\x,0) circle (2pt);

    \draw[dashdotted, thick] (4,12pt)--(4,-12pt) node[below] {$t_{k}^i$};
    \draw[dashdotted, thick] (8,12pt)--(8,-12pt) node[below] {$t_{k+1}^i$};

    \draw[<-, thick] (0,3pt)--(0,9pt) node[above] {$s_n^m$};
    \draw[<-, thick] (1.8*1.16,3pt)--(1.8*1.16,9pt) node[above] {$s_{n+1}^m$};
    \draw[<-, thick] (1.8*1.7,3pt)--(1.8*1.7,9pt) node[above=2mm] {$\ldots$};
    \draw[<-, thick] (1.8*2.3,3pt)--(1.8*2.3,9pt);
    \draw[<-, thick, BrickRed] (1.8*3.2,3pt)--(1.8*3.2,9pt);
    \draw[<-, thick] (1.8*4.1,3pt)--(1.8*4.1,9pt);

    \draw[decorate, decoration = {brace}, ultra thick] (7,-6pt)--++(-3,0) node[midway, below] {$\tau_\mathrm{MAD}^i$};

    \draw[very thick, BrickRed] (1.8*2.5,-4pt)--(1.8*2.5,9pt) node[right,rotate=45,yshift=2mm,xshift=0mm] {\color{black}$t_k^i+\Delta_k^{i,m}$};
    \draw[thick] (1.8*3.2,2.5pt)--(1.8*3.2,2.5pt) node[right,rotate=45,yshift=4mm,xshift=2mm] {$t_k^i+\overline{\Delta}_k^{i,m}$};

    \begin{axis}[
        width=7.75cm,height=0.4cm,scale only axis,
        yshift=-1.6cm,xshift=0.75cm,
        xmin=0.75,xmax=8.5,ymin=0,ymax=1,
        ylabel={$\ell^m_i$},ytick={0,1},xtick={},axis x line=none,ylabel style={rotate=-90}]
    ]
    \addplot[very thick] coordinates {
		(0.75,0)
		(4,0)
		(4,1)
		(1.8*2.5,1)
		(1.8*2.5,0)
		(8.5,0)
	};
    \end{axis}

    \begin{axis}[
        width=7.75cm,height=0.4cm,scale only axis,
        yshift=-2.3cm,xshift=0.75cm,
        xmin=0.75,xmax=8.5,ymin=0,ymax=1,
        ylabel={$b_i^m$},ytick={0,1},xtick={},axis x line=none,ylabel style={rotate=-90}]
    ]
    \addplot[very thick] coordinates {
		(0.75,0)
		(1.8*2.5,0)
		(1.8*2.5,1)
		(1.8*3.2,1)
		(1.8*3.2,0)
		(8.5,0)
	};
    \end{axis}
\end{tikzpicture}
\vspace{-2mm}
    \caption{\small Graphical representation of the indicator variables $\ell_i^m$ and $b_i^m$. Blue dots indicate $s_k^m$, $k\in\N$.}\label{fig:ellandb}
\end{figure}
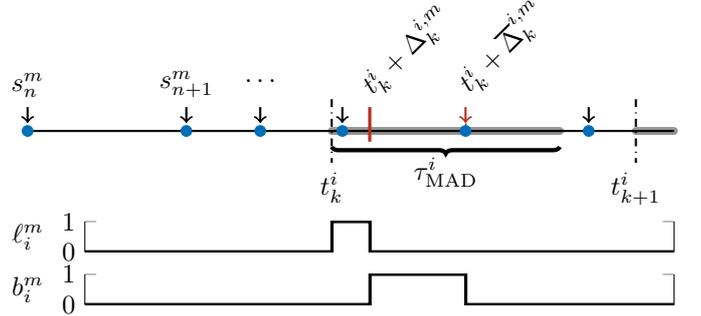

Observe in Fig. \ref{fig:ellandb}, for all $t\in[t_k^i,t_k^i+\overline{\Delta}_k^{i,m}]$, $\ell_i^m(t)+b_i^m(t)=1$. The sum of $b_i^m$ and $\ell_i^m$ indicates whether the most recently transmitted information of agent $\A_i$ has been received and processed by agent $\A_m$ ($\ell_i^m+b_i^m=0$) or that it still has to be received or processed by agent $\A_m$ ($\ell_i^m+b_i^m=1$). Moreover, due to Assumption \ref{ass:smalldelay}, both $\ell_i^m$ and $b_i^m$ are guaranteed to be zero for all $\tau\geq\tau_\text{MAD}^i$. To elucidate, we have that $\ell_i^m+b_i^m\in\{0,1\}$ for all $t\in\R_{\geq0}$. These variables are concatenated as $\ell:=(\ell_{1}^1,\ell_{2}^1,\ldots,\ell_{N}^1,\ell_{1}^2,\ell_{2}^2,\ldots,\ell_{N}^N)\in\{0,1\}^{N^2}$ and $b:=(b_{1}^1,b_{2}^1,\ldots,b_{N}^1,b_{1}^2,b_{2}^2,\ldots,b_{N}^N)\in\{0,1\}^{N^2}$.

Additionally, we define for all $i\in\Nc$, the memory variable $r_i\in\R^{n_{y,i}}$ that stores the value of the $y_i$ at the transmission times $t_k^i$, $k\in\N$. Hence, $r_i$ has dynamics $\dot{r}_i=0$ and $r_i^+=y_i$ if $\A_i$ transmits its output value. The concatenation is denoted $r:=(r_1,r_2,\ldots,r_N)\in\R^{n_y}$.

Based on the above, we can write the update of the local error $e_i^\text{in}$ during the sampling times $\{s_n^i\}_{n\in\N}$ of agent $\A_i$ as $e^\text{in}_i((s_n^i)^+)=e^\text{in}_i(s_n^i)+\diag(b^\text{in}_i(s_n^i))((r(s_n^i)-y(s_n^i))-e^\text{in}_i(s_n^i))$, where $b^\text{in}_i:=(b^i_1,b^i_2,\ldots,b^i_N)$.

\subsection{Overall system}
Using these definitions, the networked system \eqref{eq:agenti} can be cast into the form of a hybrid system $\mathcal{H}(\mathcal{C},F,\mathcal{D},G)$ with state $\xi:=\left(x,e,\tau,\sigma,r,\ell,b,{\eta}\right)\in\mathbb{X}$, where $\eta:=(\eta_{1},\eta_{2},\ldots,\eta_{N})\in\R_{\geq 0}^{N}$ and $\mathbb{X}:=\big\{ (x,e,\tau,\sigma,r,\ell,b,{\eta})\in\allowbreak\R^{n_{x}}\times\allowbreak\mathbb{E}\times\allowbreak\R_{\geq0}^{N}\times\allowbreak\mathbb{T}\allowbreak\times\R^{n_{y}}\times\{0,1\}^{N^2}\times\{0,1\}^{N^2}\times\R_{\geq0}^{N}\mid\allowbreak\forall i,m\in\Nc,((\ell_{i}^m+b_i^m=0)\lor(\ell_{i}^m+b_i^m=1\land\tau_{i}\in[0,\tau_\text{MAD}^{i}]))\big\}$.

\subsection{Flow dynamics}
We define the flow map $F:\mathbb{X}\times\R^{n_v}\to\R^{n_{x}}\times\allowbreak\R^{Nn_{y}}\times\allowbreak\{1\}^N\times\allowbreak\{1\}^N\times\allowbreak\{0\}^{n_y}\times\{0\}^{N^2}\times\{0\}^{N^2}\times\R^{N}$, as
\begin{multline}\label{eq:F(xi)}
    F(\xi,v):=\big(f(x,e,v),g(x,e,v),\mathbf{1}_{N},\mathbf{1}_{N},\mathbf{0}_{Nn_{y}},\\
    \mathbf{0}_{N^2},\mathbf{0}_{N^2},\Psi(\yh)-\varphi({\eta})\big),
\end{multline}
where the expression of $f$ follows from \eqref{eq:agenti} and \eqref{eq:networkerror}.
By combining \eqref{eq:holdingfunction} and \eqref{eq:networkerror}, we obtain that $g(x,e,v):=\allowbreak(g^1_1(x,e,v)),\allowbreak g^1_2(x,e,v)),\allowbreak \ldots,\allowbreak g^1_N(x,e,v),\allowbreak g^2_1(x,e,v),\allowbreak g^2_2(x,e,v),\allowbreak\ldots,\allowbreak g^N_N(x,e,v))$, where
\begin{equation}\label{eq:gij}
g_i^m(x,e,v):=-\delta_i(m)f_{y,i}(x,e,v)
\end{equation}
and
\begin{equation}\label{eq:fyi}
    f_{y,i}(x,e,v)=\frac{\partial h_i(x_i)}{\partial x_{i}}f_i(x,h^\text{in}_i(x)+e^\text{in}_i,v)
\end{equation}
with $\delta_i(m)$ given by $\delta_i(m)=1$ when $m\in\mathcal{V}^\text{out}_i$ and $\delta_i(m)=0$ otherwise, and $h^\text{in}_i(x):=(\delta_1(i)h_1(x_1),\delta_2(i)h_2(x_2),\ldots,\allowbreak\delta_N(i)h_N(x_N))$ with $h_i(x_i)$ in \eqref{eq:agenti}. The functions $\Psi(\yh):=(\Psi_1(\yhss{\text{in}}{1}),\allowbreak\Psi_2(\yhss{\text{in}}{2}),\allowbreak\ldots,\Psi_N(\yhss{\text{in}}{N}))$, $\varphi(\eta):=(\varphi_2(\eta_1),\varphi_2(\eta_2),\allowbreak\ldots,\allowbreak\varphi_N(\eta_N))$ with $\Psi_i:\R^{n_y}\to\R_{\geq 0}$ and $\varphi_i\in\mathcal{K}_\infty$, $i\in\Nc$, as in \eqref{eq:ETCTriggerDesign} are to be specified, as they are part of the ETM design.

The corresponding flow set is given by
\begin{equation}\label{eq:localflowset}
\mathcal{C}:=\mathbb{X}.
\end{equation}

\subsection{Jump dynamics}
To describe the jump map $G:\mathbb{X}\rightrightarrows\mathbb{X}$, we first define $\Gamma_i$ as a $N\times N$ matrix of which the $ii^\text{th}$ (diagonal) entry is equal to one and all other entries are zero, $\tilde{\Gamma}_i:=Z_i\otimes\Gamma_i$ with $Z_i:=I_N-\diag(\delta_i(1),\delta_i(2),\ldots,\allowbreak \delta_i(N))$, $\Gamma_{i,m}:=\Gamma_m\otimes\Gamma_i$ and $\Lambda_i:=\diag(\mathbf{0}_{n_{y,1}},\mathbf{0}_{n_{y,2}},\ldots,\mathbf{0}_{n_{y,i-1}},\mathbf{1}_{n_{y,i}},\mathbf{0}_{n_{y,i+1}},\ldots,\allowbreak\mathbf{0}_{n_{y,N}})$. Note that $\Gamma_{i,m}e = (0,0,\ldots,e^m_i,0,\ldots,0)$ and $\Lambda_iy=(0,0,\ldots,y_i,0,\ldots,0)$. Additionally, we define the function $\tilde{\ell}:\{0,1\}^N\to\{0,1\}$ as
\begin{equation}\label{eq:tildeli}
    \tilde{\ell}(\ell_i):=\begin{cases}0,\text{ when } \sum_{m\in\mathcal{V}^\text{out}_i}\ell^m_i=0\\
    1,\text{ when } \sum_{m\in\mathcal{V}^\text{out}_i}\ell^m_i>0\end{cases}
\end{equation}
with $\ell_i:=(\ell_i^1,\ell_i^2,\ldots,\ell_i^N)$. For the jump dynamics, we have to consider three types of jumps, namely, whether it corresponds to
\begin{enumerate}[a)]
	\item sampling instants of agent $\A_i$ with transmission (captured in $G_i^a$)
	\item sampling instants of agent $\A_i$ without transmission (captured in $G_i^b$)
	\item reception instants of information but before processing (captured in $G_{i,m}^c$).
\end{enumerate}
The jump map is thus given by $G(\xi)=\bigcup_{i\in\Nc}\bigcup_{m\in\mathcal{V}^\text{out}_i}G_{i,m}(\xi)$, where
\begin{equation}\label{eq:localjumpset}
    G_{i,m}(\xi):=\begin{cases}
    \{G_{i}^{a}(\xi)\},&\text{if } \xi\in\mathcal{D}_i\land\tilde{\ell}(\ell_i)=0~\land  \\
    &\hspace{1ex} \eta_i + \nu_i(y_i,\yhss{\text{out}}{i},\tau_i)<0\\
    \{G_{i}^{a}(\xi),G_{i}^{b}(\xi)\}, & \text{if } \xi\in\mathcal{D}_i\land\tilde{\ell}(\ell_{i})=0~\land\\
    &\hspace{1ex}\tau_i\geq\taumiet~\land\\&\hspace{1ex}\eta_i + \nu_i(y_i,\yhss{\text{out}}{i},\tau_i)=0\\
    \{G_{i}^{b}(\xi)\}, & \text{if } \tilde{\ell}(\ell_{i})=0 \land \xi\in\mathcal{D}_i~\land\\
    &\hspace{1ex} \eta_i + \nu_i(y_i,\yhss{\text{out}}{i},\tau_i)\geq0\\
    \{G_{i,m}^c(\xi)\}, & \text{if } \xi\in\mathcal{D}_i\land\ell_i^m=1\\
    \emptyset, & \text{if }\xi\notin \mathcal{D}_{i}
\end{cases}
\end{equation}
where
\begin{multline}
    G_{i}^{a}(\xi):=\big(x,e+\diag(b)(\Gamma_i\otimes I_{n_y})(\mathbf{1}_N\otimes(r-y)-e),\\
    (I_{N}-\Gamma_{i})\tau,(I_{N}-\Gamma_{i})\sigma,\Lambda_iy+(I_{n_y}-\Lambda_i)r,\\
    \ell+\tilde{\Gamma}_i\mathbf{1}_{N^2},(I_{Nn_y}-\Gamma_i\otimes I_{n_y})b,\Gamma_{i}\varrho_i(e^\text{out}_i)+{\eta}\big),
\end{multline}
that corresponds to case a) above,
\begin{multline}
    G_{i}^{b}(\xi):=\big(x,e+\diag(b)(\Gamma_i\otimes I_{n_y})(\mathbf{1}_N\otimes(r-y)-e),\\
    \hspace{4mm}\tau,(I_{N}-\Gamma_{i})\sigma,r,\ell,(I_{Nn_y}-\Gamma_i\otimes I_{n_y})b,\\
    \Gamma_{i}\nu_i(e^\text{out}_i,\tau_i)+{\eta}\big)
\end{multline}
for case b) above, and
\begin{multline}
    G_{i,m}^c(\xi):=\big(x,e,\tau,\sigma,r,\ell-\Gamma_{i,m}\mathbf{1}_{N^2},\\
    b+\Gamma_{i,m}\mathbf{1}_{N^2},{\eta}\big),
\end{multline}
for case c).

The functions $\varrho_i:\R^{n_{y,i}}\times\R^{Nn_{y,i}}\to\R_{\geq0}$ and $\nu_i:\R^{n_{y,i}}\times\R^{Nn_{y,i}}\times\R_{\geq0}\rightrightarrows\R_{\leq0}$, are to be designed, as part of the triggering mechanism, see \eqref{eq:ETCTriggerDesign}. When $b_i^m=1$ for some $m\in\mathcal{V}^\text{out}_i$, for both case a) and b) above, $b$ is used to update the local error $e^\text{in}_i$. Furthermore, after a sampling instant of agent $\A_i$, $b_i^m=0$ for all $m\in\Nc$.

To complete the description of the jump map we also have to define the sets $\mathcal{D}_{i}$, which we will do next.

The corresponding jump set $\mathcal{D}\subseteq\mathbb{X}$ is given by $\mathcal{D}:=\bigcup_{i\in\Nc}\mathcal{D}_{i}$, where
\begin{equation}
    \mathcal{D}_{i} :=\left\{ \xi\in\mathbb{X}\mid\sigma_i\geq d_i\lor \tilde{\ell}(\ell_{i})=1 \right\}\label{eq:localjumpset2}
\end{equation}
with $d_i$ from \eqref{eq:masp}.

Observe that a jump is enforced when $\ell_i^m=1 \lor b_i^m=1$ and $\tau_i=\tau_\text{MAD}^i$, or when $\sigma_i=\tau^i_\text{MASP}$. As such, the hybrid model complies with Assumption \ref{ass:smalldelay} and with \eqref{eq:masp}.

\subsection{Formal problem statement}\label{sec:problemStatement}
We can now state the problem of Section \ref{sec:objective} formally as follows.
\begin{problem}\label{problem}
    \label{Problem}Given the system $\mathcal{H}(\mathcal{C},F,\mathcal{D},G)$, provide design conditions for the time-constants $\tau_\text{MAD}^i,\tau_\text{MIET}^i\in\R_{>0}$ with $\tau_\text{MIET}^i\geq\tau_\text{MAD}^i$ and the functions $\Psi_i$, $\varsigma_i$, $\varrho_i$ and $\nu_i$ as in \eqref{eq:ETCCondition} and \eqref{eq:ETCTriggerDesign}, for $i\in\Nc$, such that, under Assumption 1, the system $\mathcal{H}$ is persistently flowing\footnote{Persistently flowing in the sense that maximal solutions have an unbounded domain in $t$-direction, see \cite{Goebel_Sanfelice_Teel_2012}.} and ($\tilde{s},\mathcal{S}$)-flow-dissipative for a set $\mathcal{S}\subset\mathbb{X}$, for a given supply rate $\tilde{s}:\mathbb{X}\times\R^{n_v}\to\R$ of the form
    \begin{equation}\label{eq:supplyrate}
        \tilde{s}(\xi,v):=s(x,e,v)-\varphi({\eta}),
    \end{equation}
    where $\xi\in\mathbb{X}$, $v\in\R^{n_v}$ and $\varphi:=(\varphi_1(\eta_1),\varphi_2(\eta_2),\ldots,\varphi_N(\eta_N))$ with $\varphi_i$ as in \eqref{eq:ETCTriggerDesign}.\endstatement
\end{problem}
As shown in, for example, \cite{Scha96,Teel_2010}, the use of dissipativity allows the consideration of various important system properties such as asymptotic stability, input-to-state stability, $\mathcal{L}_p$-stability with $p\in[1,\infty)$ and passivity, from a unified point of view.
Thus, the supply rate $\tilde{s}$ and the set $\mathcal{S}$ capture the desired stability, performance and robustness requirements.

\section{Design conditions}\label{sec:DesignConditions}
To ensure that the hybrid system has the desired performance and stability properties, the following conditions have to be satisfied.

\subsection{Growth of the network-induced error}
We require that the dynamics of the network-induced error satisfy the following property.
\begin{condition}\label{cond:W3}
    For each $i\in\Nc$, there exist functions $H_i:\R^{n_x}\times\R^{n_y}\times\R^{n_v}\to\R_{\geq 0}$ and constants $L_i\geq0$ such that for all $m\in\mathcal{V}^\text{out}_i$, $x\in\R^{n_x}$, $e\in\R^{Nn_y}$ and $v\in\R^{n_v}$,
    \begin{equation}\label{eq:bargibound}
        \vert f_{y,i}(x,e,v)\vert\leq H_{i}(x,e^\text{in}_i,v)+L_{i}\vert e^i_i\vert,
    \end{equation}
    where $f_{y,i}(x,e,v)=\frac{\partial h_i(x_i)}{\partial x_{i}}f_i(x,h^\text{in}_i(x)+e^\text{in}_i,v)$ with $\delta_i(m)$ given by $\delta_i(m)=1$ when $m\in\mathcal{V}^\text{out}_i$ and $\delta_i(m)=0$ otherwise, and $h^\text{in}_i(x):=(\delta_1(i)h_1(x_1),\delta_2(i)h_2(x_2),\ldots,\allowbreak\delta_N(i)h_N(x_N))$ with $h_i(x_i)$ in \eqref{eq:agenti}. \endstatement
\end{condition}
Inequality \eqref{eq:bargibound} is related to $\dot{y}_i$, which, due to the use of ZOH devices, is directly related to $\dot{e}^m_i$, as $\dot{e}^m_i=-\delta_i(m)\dot{y}_i$. In essence, Condition \ref{cond:W3} is a restriction on the growth of the network-induced error between transmissions. This condition is naturally satisfied for linear systems or when the vector fields $f_i$ are globally bounded and $h_i$ are globally Lipschitz.
\subsection{Lower-bounds on the Minimum Inter-event Times and Maximum Allowable Delays}
To obtain lower-bounds on the minimum inter-event times $\taumiet$ and the maximum allowable delay $\tau_\text{MAD}^i$ for each agent $\A_i$, $i\in\Nc$, we first characterize the influence of the transmission errors $e^\text{in}_i$ on the state $x$ and the desired stability/performance property by means of the following condition.
\begin{condition}\label{cond:V}
    There exist a locally Lipschitz function $V:\R^{n_x}\to\R_{\geq 0}$ and a non-empty closed set $\mathcal{X}\subseteq \R^{n_x}$, $\mathcal{K}_\infty$-functions $\underline{\alpha}_{V}\leq\overline{\alpha}_{V}$, continuous functions $\varsigma_{i}:\R^{Nn_y}\to\R_{\geq 0}$, constants $\mu_i,\gamma_{i}>0$, $i\in\Nc$, such that for all $x\in\R^{n_x}$
    \begin{equation}\label{eq:BoundsVdelayfree}
        \underline{\alpha}_{V}(\vert x\vert_{\mathcal{X}} )  \leq  V(x) \leq\overline{\alpha}_{V}(\vert x\vert_{\mathcal{X}} ),
    \end{equation}
    and for all $y\in\R^{n_y}$, $e\in\R^{Nn_y}$, $v\in\R^{n_v}$, and almost all $x\in\R^{n_x}$
    \begin{multline}\label{eq:Vflowcondition}
        \textstyle\hspace{-4mm}\left\langle \nabla V(x),f(x,e,v)\right\rangle  \leq s(x,e,v)-\sum_{i\in\Nc}\varsigma_{i}(\yhss{\text{in}}{i})\\
        \textstyle+\sum_{i\in\Nc}\big(-\mu_iN_iH_{i}^{2}(x,e^\text{in}_i,v)+\gamma_{i}^2 \vert e^\text{out}_i\vert^2\big)
    \end{multline}
    with $N_i$ the cardinality of $\mathcal{V}^\text{out}_i$, and $H_i$ from \eqref{eq:bargibound}.\endstatement
\end{condition}
Condition \ref{cond:V} constitutes an $\mathcal{L}_2$-gain condition from $|e_i^\text{out}|$ to $H_i$. In case of a linear system, this condition can always be verified if the system is controllable, for instance. In the absence of a network, i.e., when $e=0$, \eqref{eq:BoundsVdelayfree}-\eqref{eq:Vflowcondition} imply an $(s,\mathcal{X})$-dissipativity property. However, this property is affected by the network-induced error $e$, and our objective is to design the triggering mechanisms in such a way that the dissipativity property still holds for the networked system.

The constants $\gamma_{i}$ as in Condition \ref{cond:V} are used to determine $\tau_\text{MIET}^{i}$ and $\tau_\text{MAD}^{i}$, $i\in\Nc$, via the following condition.

\begin{condition}\label{cond:mietmad}
    Select $\tau_\text{max}^i>0$ and $\tau_\text{MAD}^{i}>0$, $i\in\Nc$,
    with $\tau_\text{max}^{i}\geq\tau_\text{MAD}^{i}+\tau_\text{MASP}^{i}$ such that
    \begin{align}
        \hspace{-1.75mm}\tilde{\gamma}_{i}(0)\phi_{0,i}(\tau_\text{max}^{i}) & \geq \lambda_{i}^{2}\tilde{\gamma}_{i}(1)\phi_{1,i}(0),\label{eq:tauimiet}\\
        \tilde{\gamma}_{i}(1)\phi_{1,i}(\tau_{i}) & \geq  \tilde{\gamma}_{i}(0)\phi_{0,i}(\tau_{i}),\text{ for all }\tau_i\in[0,\tau_\text{MAD}^{i}],\label{eq:tauimad}
    \end{align}
    where $\phi_{l,i}$, $l\in\{0,1\}$, evolves according to
    \begin{gather}
        \tfrac{d}{d\tau_i}\phi_{l,i}=-\left(
        2\tilde{L}_i(l)\phi_{l,i}+\tilde{\gamma_{i}}(l)\left(\tfrac{1}{\mu_i\epsilon_i}\phi_{l,i}^{2}+1\right)\right),\label{eq:diffeqphi01}
    \end{gather}
    for some fixed initial conditions $\phi_{l,i}(0)$, $l\in\{0,1\}$, that satisfy
    $\tilde{\gamma}_{i}(1)\phi_{1,i}(0)\geq\tilde{\gamma}_{i}(0)\phi_{0,i}(0)>\lambda_{i}^{2}\tilde{\gamma}_{i}(1)\phi_{1,i}(0)>0$,
    where, for each $i\in\Nc$ and $l\in\{0,1\}$, the functions $\tilde{L}_i:\{0,1\}\to\R_{\geq 0}$
    and $\tilde{\gamma}_{i}:\{0,1\}\to\R_{\geq 0}$ are given by
    \begin{align}
        \tilde{L}_i(l):=\lambda_i^{-l}\sqrt{N_i}L_i,\qquad
        \tilde{\gamma}_i(l):=\lambda_i^{-l}\gamma_{i}\label{eq:gamma01},
    \end{align}
    with $N_i$ the cardinality of $\mathcal{V}^\text{out}_i$ and where $\mu_i$ and $\gamma_{i}$ satisfy Condition \ref{cond:V}. The constants $\lambda_{i}\in(0,1)$ and $\epsilon_i\in(0,1]$, $i\in\Nc$, are tuning parameters. If the above conditions are satisfied, $\tau_\text{MIET}^i$ is defined as $\tau_\text{MIET}^i:=\tau_\text{max}^i-\tau_\text{MASP}^i$.     \endstatement
\end{condition}
Condition \ref{cond:mietmad} can always be ensured, as long as sufficiently fast sampling is available. In practice, based on the constants $\gamma_i$, $(\tau_\text{max}^i,\tau_\text{MAD}^i)$ curves can be generated to intuitively select appropriate values for $\lambda_i$, $\phi_{0,i}(0)$ and $\phi_{1,i}(0)$.

These conditions are similar to the conditions in \cite{Dolk_Borgers_Heemels_2017}, even though PETC or the effect of sampling is not considered. Indeed, in the continuous-time case, i.e., when $\tau_\text{MASP}$ approaches zero, $\tau_\text{MIET}^i=\tau_\text{max}^i$. This statement underlines that, if faster sampling is used, the continuous-time ETC behavior is recovered in the proposed setup.

\subsection{Event-triggering Mechanism Design}\label{sec:ETMdesign}
To facilitate the design of the ETM, consider the following condition.

\begin{condition}\label{cond:H}
    For $i\in\Nc$, consider the function $H_i$ satisfying Condition \ref{cond:V}. There exist locally Lipschitz functions $\underline{H}_i:\R^{n_y}\to\R_{\geq 0}$ that for all $e\in\R^{Nn_y}$, $v\in\R^{n_v}$ and $x\in\R^{n_x}$, satisfy $\underline{H}_i(\yhss{\text{in}}{i})\leq H_i(x,e^\text{in}_i,v)$.\endstatement
\end{condition}
The function $\Psi_{i}$ in \eqref{eq:etadot} is given by, for any $\yhss{\text{in}}{i}$,
\begin{equation}
    \label{eq:Psi2}\Psi_{i}(\yhss{\text{in}}{i}):=\varsigma_i(\yhss{\text{in}}{i})+(1-\epsilon_i)\mu_iN_i\underline{H}_i^2(\yhss{\text{in}}{i}),
\end{equation}
where $\varsigma_i$ and $\underline{H}_i$ come from Conditions \ref{cond:V} and \ref{cond:H}, respectively.
The function $\varrho_{i}$ is given by, for any $y_i\in\R^{n_{y,i}}$ and $\yhss{\text{out}}{i}\in\R^{Nn_{y,i}}$
\begin{equation}\label{eq:etai0}
    \varrho_i(y_{i},\yhss{\text{out}}{i}):=\varepsilon_\varrho\vert e^\text{out}_i\vert^2
\end{equation}
with $\varepsilon_\varrho:=\left(\tilde{\gamma}_{i}(0)\phi_{0,i}(\tau_\text{MIET}^i+\sigma_i)-\tilde{\gamma}_{i}(1)\phi_{1,i}(0)\lambda_i^2\right)$ where $\phi_{l,i}$, $l\in\{0,1\}$ as in \eqref{eq:diffeqphi01} and $\tilde{\gamma_i}:\{0,1\}\to\R$ is as in \eqref{eq:gamma01}.
Finally, the function $\nu_{i}:\R^{n_y}\times \R_{\geq 0}\rightrightarrows\R_{\leq0}$ is defined as
\begin{equation}\label{eq:etai1}
    \nu_i(y_i,\yhss{\text{out}}{i},\tau_i):= (1-\omega_i(\tau_i))\tilde{\gamma}_i(0)\varepsilon_\nu\vert e^\text{out}_i\vert^2,
\end{equation}
where $\varepsilon_\nu:=-\left(\phi_{0,i}(\tau_\text{MIET}^i)-\phi_{0,i}(\tau_\text{MIET}^i+\sigma_i)\right)$ and
\begin{equation}
    \omega_{i}(\tau_i) \in \begin{cases}
    \{1\}, & \text{for }\tau_i\in[0,\taumiet) \\
    [0,1], & \text{for }\tau_i=\taumiet,\\
    \{0\}, & \text{for }\tau_i >\taumiet.
\end{cases}\label{eq:Omegai}
\end{equation}
Note that $\nu_i$ is single-valued for all $\tau_i\neq\tau_\text{MIET}^i$, and set-valued for $\tau_i=\tau_\text{MIET}^i$. Since the proof holds for \emph{all} points in the set-valued map, in essence we can use the discontinuous version ($\omega_i(\tau_i)=1$ if $\tau_i\leq\taumiet$ and 0 otherwise) to verify the condition in \eqref{eq:ETCCondition}. Hence, the fact that $\nu_i$ is set-valued is not an issue with respect to \eqref{eq:ETCCondition}.

In the proposed setup, each agent needs to know (and compute) constants $\varepsilon_\varrho$ and $\varepsilon_\nu$ on-line due to the dependence on $\sigma_i$. If, from a computational standpoint, this is infeasible, a conservative upper-bound can be used by taking $\varepsilon_\varrho:=\left(\tilde{\gamma}_{i}(0)\phi_{0,i}(\tau_\text{max}^i)-\tilde{\gamma}_{i}(1)\phi_{1,i}(0)\lambda_i^2\right)$ and $\varepsilon_\nu:=\left(\phi_{0,i}(\tau_\text{max}^i)-\phi_{0,i}(\tau_\text{MIET}^i)\right)$, which can be computed a priori.

We emphasize that the local ETMs as described by \eqref{eq:ETCCondition}, \eqref{eq:ETCTriggerDesign}, \eqref{eq:Psi2}, \eqref{eq:etai0} and \eqref{eq:etai1}, can operate fully asynchronously in the sense that clock synchronization or acknowledgment signals are not required.

\subsection{Main result}\label{sec:MainResult}
Given the ETM design and the corresponding hybrid model presented above, we can now state the following result. Its proof is provided in the appendix.

\begin{thm}\label{thm:Theorem}
    Consider the system $\mathcal{H}(\mathcal{C},F,\mathcal{D},G)$ where ${\Psi}_i$, $\varrho_i$ and $\nu_i$ are given by \eqref{eq:Psi2}, \eqref{eq:etai0} and \eqref{eq:etai1}, respectively. Moreover, suppose that Conditions \ref{cond:W3}-\ref{cond:H} hold. Then the MAS described by ${\mathcal{H}}$ is $(\tilde{s}, \mathcal{S})$-flow-dissipative with the supply rate $\tilde{s}:\mathbb{X}\times\R^{n_v}\to\R$ as given in \eqref{eq:supplyrate} and $\mathcal{S}=\allowbreak\{\xi \in \mathbb{X} \mid x\in\mathcal{X}, \ \allowbreak e=0, \ \allowbreak \eta=0\}.$ In addition, if {there are no finite escape times during the flow\footnote{The absence of finite escape times during flow is meant here in the sense that case (b) in Prop.~2.10 in \cite{Goebel_Sanfelice_Teel_2012} cannot occur.}, then the system $\mathcal{H}$ is persistently flowing.}\endstatement
\end{thm}
Theorem \ref{thm:Theorem} implies that the desired stability and/or performance properties, guaranteed by the local controllers in absence of the network, are preserved by the original dissipativity property when the network is taken into account.

\section{Case study}\label{sec:casestudy}
We apply the results to the single-integrator consensus problem, where we have a multi-agent system with $N\in\N_{>0}$ agents. All agents have state $x_i\in\mathbb{R}$, $i\in\mathcal{N}$, whose dynamics evolve according to $\dot{x}_i=u_i$ with $u_i\in\R$ the control input. The output of the system is the state $x_i$, i.e., $y_i=x_i$. We assume that the graph $\mathcal{G}(\mathcal{V},\mathcal{E})$ with Laplacian matrix $L$ is connected and undirected, i.e., $L^\top=L$. The control objective is for the states of all agents to asymptotically converge, i.e., $\lim_{t\to\infty}|x_i(t)-x_m(t)|=0$ for all $i,m\in\mathcal{N}$. To achieve consensus, we implement the control law
\begin{equation}\label{eq:consensusui}
    \textstyle u_i=-\sum_{m\in\mathcal{V}_i^\text{in}}(x_i+e_i^i-x_m-e_m^i).
\end{equation}
We consider the Lyapunov candidate $V(x)=x^\top L x$ where $x:=(x_1,x_2,\ldots,x_N)$. According to \cite{Dolk_Postoyan_Heemels_2019}, the derivative of this Lyapunov function can be upper-bounded as $\left\langle \nabla V(x),-Lx-Le\right\rangle\allowbreak\leq\allowbreak\sum_{i\in\Nc}\big(-d_iz_i^2-c_iu_i^2+\allowbreak(\gamma_i^2-\alpha_i)|e^i_i|^2\big)$ with $d_i:=\delta(1-aN_i)$, $c_i:=(1-\delta)(1-aN_i)$ and $\gamma_i=\sqrt{{a^{-1}N_i+\alpha_i}}$, and where $\delta\in(0,1)$,  $a\in(0,\frac{1}{N_i})$ and $\alpha_i>0$ are tuning parameters. The theorem below shows the exact expressions for all the required conditions. Its proof is omitted for space reasons.
\begin{thm}\label{thm:consensus}
    The system with local dynamics $\dot{x}_i=u_i$ and local controller \eqref{eq:consensusui} satisfies Conditions \ref{cond:W3}, \ref{cond:V} and \ref{cond:H} with $H_i=|u_i|$, $L_i=0$, $s(x,e)=\sum_{i\in\Nc}\left(-d_iz_i^2-\mu_ie_i^2\right)$, $\mathcal{X}=\{x\in\R^{N}\mid x_1=x_2=\ldots=x_N\}$, $\varsigma_i=0$, $\mu_i=c_i\frac{1}{N_i}$, $\gamma_{i}=\sqrt{{a}^{-1}N_i+\alpha_i}$, and $\underline{H}_i=|u_i|$. \endstatement
\end{thm}
Constants $\tau_\text{max}^i$ and $\tau_\text{MAD}^i$ can be generated via an intuitive procedure, as described in \cite{Dolk_Postoyan_Heemels_2019}. Theorem \ref{thm:consensus} implies that asymptotic consensus is achieved with the proposed control configurations in this paper.

We simulate the same system as \cite{Dolk_Postoyan_Heemels_2019} with non-uniform and time-varying transmission delays. However, in our case we implement our periodic event-triggered control framework instead of continuous-event triggered control as in \cite{Dolk_Postoyan_Heemels_2019}. The system has $N=8$ agents which are connected as described by a graph $\mathcal{G}$ with undirected edges $(1,2)$, $(1,8)$, $(2,3)$, $(2,7)$, $(3,4)$, $(3,6)$, $(4,5)$, $(5,6)$, $(5,8)$ and $(7,8)$. We use the parameters $\delta=\alpha_i=0.05$, $a=0.1$ and $\epsilon_i=0.5$ for all $i\in\mathcal{N}$. Given these tuning parameters, we obtain $\gamma_i=4.478$ and $c_i=0.76$ for agents $i\in\mathcal{N}$ with two neighbors (i.e., $N_i=2$, thus agents $P_1$, $P_4$, $P_6$ and $P_7$) and $\gamma_i=5.482$ and $c_i=0.665$ for agents $i\in\mathcal{N}$ with three neighbors (i.e., $N_i=3$, thus agents $P_2$, $P_3$, $P_5$ and $P_8$). The function $\varphi_i(\eta_i)$ is designed as $\varphi_i(\eta_i)=-\epsilon_\eta(\eta_i)$ with $\epsilon_{\eta}=0.05$. We select $\lambda_i=0.2$ for all agents, and pick $\phi_{0,i}(0)=5$ and $\phi_{1,i}=2$.
For these values, we obtain $(\tau^i_\text{max},\tau_\text{MAD}^i)=(0.12,0.016)$ for agents $i\in\mathcal{N}$ for which $N_i=2$ and $(\tau^i_\text{max},\tau_\text{MAD}^i)=(0.09,0.012)$ for agents $i\in\mathcal{N}$ for which $N_i=3$. We select $\tau_\text{MIET}^i=0.07$ ($\tau_\text{MIET}^i=0.05$) for all agents for which $N_i=2$ ($N_i=3$), respectively,  $\tau_\text{MASP}^i=10^{-2}$ and $d_i=10^{-3}$ for all $i\in\mathcal{N}$. At each sampling moment $s_n^i$, the next sampling moment is scheduled randomly such that $s_{n+1}^i\in[s_n^i+d_i,s_n^i+\tau_\text{MASP}^i]$ for each $i\in\Nc$, hence the sampling of each agent is aperiodic, asynchronous and independent of the other agents. The state evolution and inter-event times are depicted in Fig. \ref{fig:simulation}, confirming our main theorem.

\begin{figure}
    \centering
    \input{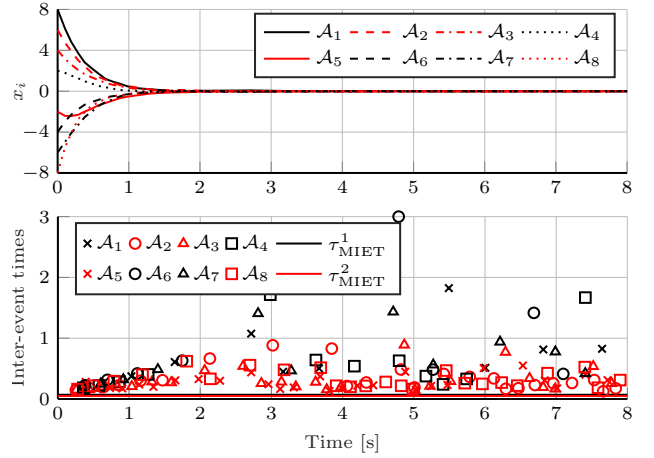}
    \caption{\small States and inter-event times for the example in Section \ref{sec:casestudy}.}
    \label{fig:simulation}
\end{figure}

\section{Conclusion}
We presented a framework for the design of Zeno-free dynamic periodic triggering mechanisms for the control of nonlinear multi-agent systems via a packet-based network. The method can cope with non-uniform and time-varying delays. By ensuring that the conditions of the local trigger mechanisms only have to be verified at the local (asynchronous) sampling times, the proposed framework is suitable for implementation on digital platforms. With a dissipativity property, the framework can handle several relevant stability and performance properties such as asymptotic (set) stability, input-to-state stability, $\mathcal{L}_p$-stability with $p\in[1,\infty)$ and consensus, from a unified point of view. Thereby, capturing a wide range of systems to which these techniques can be applied.

\bibliography{ifacconf}

\appendix
\section{Proof of Theorem \ref{thm:Theorem}.}
To verify the dissipativity properties of the digitally implemented MAS with respect to the supply rate $\tilde{s}(\xi,w)$, we aim to construct a storage function that satisfies Definition \ref{defn:Vdiss}.

For the clarity of exposition, the proof is composed of several steps. Firstly, we present a candidate storage function. Next, we prove several intermediate results. Finally, we show that the overall candidate storage function indeed satisfies Definition \ref{defn:Vdiss}.

\textbf{Step I.} \emph{Candidate storage function.}
Consider the following candidate storage function
\begin{multline}\label{eq:LyapFunction2}
U(\xi)=V(x)+\sum_{i\in\Nc}\eta_{i}\\+\sum_{i\in\Nc}\tilde{\gamma}_{i}(p(\ell_{i},b_i))\bar{\phi}_{p(\ell_{i},b_i),i}(\tau_i,\sigma_i)\tilde{W}_{i}^{2}(\ell_{i},b_i,y_i,e^\text{out}_i,r_{i}),
\end{multline}
for any $\xi\in\mathbb{X}$ with $\ell_i=(\ell_i^1,\ell_i^2,\ldots \ell_i^N)\in\{0,1\}^N$, $b_i=(b_i^1,b_i^2,\ldots,b_i^N)\in\{0,1\}^N$ and the function $p:\{0,1\}^N\times\{0,1\}^N\to\{0,1\}$ as
\begin{equation}
    p(\ell_i,b_i):=\begin{cases}
    0,&\text{when }\sum_{m\in\mathcal{V}^\text{out}_i}\ell_i^m+b_i^m=0,\\
    1,&\text{when }\sum_{m\in\mathcal{V}^\text{out}_i}\ell_i^m+b_i^m>0,
    \end{cases}
\end{equation}
and where the function $V:\R^{n_x}\to\R_{\geq0}$ satisfies \eqref{eq:Vflowcondition}, the function $\bar{\phi}_{l,i}:\R_{\geq 0}\times\mathbb{T}_i\to\R_{\geq 0}$, $l\in\{0,1\}$, $i\in\Nc$ is given by
\begin{equation}\label{eq:barphi}
    \bar{\phi}_{l,i}(\tau_i,\sigma_i):=\begin{cases}
    \phi_{l,i}(\tau_i),&\text{ when }\tau_i-\sigma_i\leq\tau_\text{MIET}^i,\\
    \phi_{l,i}(\tau_\text{MIET}^i+\sigma_i),&\text{ when }\tau_i-\sigma_i>\tau_\text{MIET}^i,
\end{cases}
\end{equation}
where $\phi_{l,i}$ evolves as in \eqref{eq:diffeqphi01}. The function $\tilde{W}_i:\{0,1\}^N\times\{0,1\}^N\times\R^{Nn_{y,i}}\times\R^{Nn_{y,i}}\times\R^{n_{y,i}}\to\R_{\geq0}$ is given by
\begin{multline}\label{eq:Wtilde}
    \tilde{W}_{i}\left(\ell_i,b_i,y_i,e^\text{out}_i,r_{i}\right) := \max\Bigg\{\left\vert e^\text{out}_i+s_i(\ell_i,b_i,y_i,e^\text{out}_i,r_{i})\right\vert,\\  \lambda_i\max_{\mathcal{R}\subset\mathcal{R}_i(\ell_i,b_i)}\Big\vert e^\text{out}_i+\sum_{l\in\mathcal{R}}Y_ls_i(\ell_i,b_i,y_i,e^\text{out}_i,r_{i})\Big\vert\Bigg\},
\end{multline}where $Y_l:=(\Gamma_l\otimes I_{n_{y,i}})$,
$\mathcal{R}_i(\ell_i,b_i):=\{m\in\mathcal{V}^\text{out}_i\mid \ell_i^m=1\lor b_i^m=1\}$, which is the set of agents that still have to receive or process the latest transmitted measurement by agent $\A_i$ and where
\begin{equation}\label{eq:barsi}
    s_i(\ell_i,b_i,y_i,e^\text{out}_i,r_i)=\hspace{-1ex}\sum_{l\in\mathcal{R}_i(\ell_i,b_i)}\hspace{-1ex}Y_l(\mathbf{1}_N\otimes(r_i-y_i)-e^\text{out}_i)
\end{equation}
with the variables $\tilde{\gamma}_{i}(l)\in\R_{\geq0}$, $l\in\{0,1\}$, as in \eqref{eq:gamma01}.

\textbf{Step II.} \emph{Properties of $\tilde{W}_{i}$ and $V$.}
Below, we present some intermediate results in Lemma \ref{lemma:Wjump}-\ref{lemma:etajump}.
\begin{lemma}\label{lemma:Wjump}
	Consider the function $\tilde{W}_{i}$ defined \eqref{eq:Wtilde}. For each $i\in\Nc$, $m\in\mathcal{V}^\text{out}_i$ and for all $e^\text{out}_i\in\R^{Nn_{y,i}}$, $r_i\in\R^{n_{y,i}}$ and $0<\lambda_{i}<1$, the function $\tilde{W}_{i}$ satisfies for
    \underline{{\color{blue}update} events:}
    \begin{multline}
        \tilde{W}_{i}(\ell_{i}-\Gamma_m\mathbf{1}_N,b_i+\Gamma_m\mathbf{1}_N,y_i,e^\text{out}_i,r_{i})\\
        =\tilde{W}_{i}(\ell_{i},b_i,y_i,e^\text{out}_i,r_{i}),\label{eq:WleqW1}
    \end{multline}
	\underline{{\color{blue}sampling} events without transmission:}
	\begin{multline}
    	\tilde{W}_{i}(\ell_i,b_i-\Gamma_m\mathbf{1}_N,y_i,e^\text{out}_i+Y_ms_i(\ell_i,b_i,y_i,e^\text{out}_i,r_{i}),r_{i})\\
        \leq\tilde{W}_{i}(\ell_i,b_i,y_i,e^\text{out}_i,r_{i}),\label{eq:WleqW2}
	\end{multline}
	and \underline{{\color{blue}sampling} events with transmission:}
	\begin{multline}
    	\tilde{W}_{i}(Z_i\mathbf{1}_N,\mathbf{0}_N,y_i,e^\text{out}_i,y_i)\leq\lambda_{i}\tilde{W}_{i}(\mathbf{0}_N,\mathbf{0}_N,y_i,e^\text{out}_i,r_{i}),\label{eq:WleqW3}
    \end{multline}
	with $Y_m=(\Gamma_m\otimes I_{n_y})$ as before and where $\bar{Y}_m:=(I_{Nn_y}-\Gamma_m\otimes I_{n_y})$.\endstatement
\end{lemma}
\begin{proof}
    \underline{Equality \eqref{eq:WleqW1}.} The equality follows directly from the fact that, for each $i\in\Nc$, $m\in\mathcal{V}^\text{out}_i$, $\mathcal{R}_i(\ell_i-\Gamma_m\mathbf{1}_N,b_i+\Gamma_m\mathbf{1}_N)=\mathcal{R}_i(\ell_i,b_i)$ and the definition of $s_i$ in \eqref{eq:barsi}.

    \underline{Inequality \eqref{eq:WleqW2}.} The first part of \eqref{eq:WleqW2} with $\tilde{W}_i$ as in \eqref{eq:Wtilde} is equal to
	\begin{multline}\label{eq:tildeWi0}
    	\tilde{W}_{i}(\ell_i,b_i-\Gamma_m\mathbf{1}_N,y_i,e^\text{out}_i+Y_ms_i(\ell_i,b_i,y_i,e^\text{out}_i,r_{i}),r_{i})\\
        =\max\Big\{|e^\text{out}_i+Y_ms_i(\ell_i,b_i,y_i,e^\text{out}_i,r_i)\hspace{30mm}\\
        \hspace{10mm}+s_i(\ell_i,b_i-\Gamma_m\mathbf{1}_N,y_i,e^\text{out}_i+Y_ms_i(\ell_i,b_i,y_i,e^\text{out}_i,r_i),r_i)|,\\
        \lambda_i\max_{\mathcal{R}\subset\mathcal{R}_i(\ell_i,b_i-\Gamma_m\mathbf{1}_N)}\Big|e^\text{out}_i+Y_ms_i(\ell_i,b_i,y_i,e^\text{out}_i,r_i)\\
        +\sum_{l\in\mathcal{R}}Y_ls_i(\ell_i,b_i-\Gamma_m\mathbf{1}_N,y_i,e^\text{out}_i+Y_ms_i(\ell_i,b_i,y_i,e^\text{out}_i,r_i),r_i)\Big|\Big\}
	\end{multline}
    for each $i\in\Nc$, $m\in\mathcal{V}^\text{out}_i$. By means of \eqref{eq:barsi} we find that
    \begin{multline}\label{eq:tildeWi1}
        e^\text{out}_i+s_i(\ell_i,b_i,y_i,e^\text{out}_i,r_i)=e^\text{out}_i+Y_ms_i(\ell_i,b_i,y_i,e^\text{out}_i,r_i)\hspace{20mm}\\
        \hspace{10mm}+s_i(\ell_i,b_i-\Gamma_m\mathbf{1}_N,y_i,e^\text{out}_i+Y_ms_i(\ell_i,b_i,y_i,e^\text{out}_i,r_i),r_i).
    \end{multline}
    Moreover, given the fact that $Y_lY_m=0$ for $l\neq m$, we have
    \begin{multline}\label{eq:tildeWi2}
        \max_{\mathcal{R}\subset\mathcal{R}_i(\ell_i,b_i-\Gamma_m\mathbf{1}_N)}\Big|e^\text{out}_i+Y_ms_i(\ell_i,b_i,y_i,e^\text{out}_i,r_i)\\
        +\sum_{l\in\mathcal{R}}Y_ls_i(\ell_i,b_i-\Gamma_m\mathbf{1}_N,y_i,e^\text{out}_i+Y_ms_i(\ell_i,b_i,y_i,e^\text{out}_i,r_i),r_i)\Big|\\
        =\max_{\mathcal{R}\subset\mathcal{R}_i(\ell_i,b_i-\Gamma_m\mathbf{1}_N)}\Big|e^\text{out}_i+\sum_{l\in\mathcal{R}\cup\{m\}}Y_ls_i(\ell_i,b_i,y_i,e^\text{out}_i,r_i)\Big|\\
        \leq\max_{\mathcal{R}\subset\mathcal{R}_i(\ell_i,b_i)}\Big|e^\text{out}_i+\sum_{l\in\mathcal{R}}Y_ls_i(\ell_i,b_i,y_i,e^\text{out}_i,r_i)\Big|.
    \end{multline}
    Combining \eqref{eq:tildeWi1} and \eqref{eq:tildeWi2} with \eqref{eq:tildeWi0}, we obtain \eqref{eq:WleqW2}.

    \underline{Inequality \eqref{eq:WleqW3}.} Observe that \eqref{eq:WleqW3} with $\tilde{W}_i$ as in \eqref{eq:Wtilde} is equal to
    \begin{multline}\label{eq:WleqW3_a}
        \max\Big\{|e^\text{out}_i+s_i(Z_i\mathbf{1}_N,\mathbf{0}_N,y_i,e^\text{out}_i,y_i)|,\\
        \lambda_i\max_{\mathcal{R}\subset\mathcal{R}_i(Z_i\mathbf{1}_N,\mathbf{0}_N)}\Big|e^\text{out}_i+\sum_{l\in\mathcal{R}}Y_ls_i(Z_i\mathbf{1}_N,\mathbf{0}_N,y_i,e^\text{out}_i,y_i)\Big|\Big\}\\
        \leq\lambda_i|e^\text{out}_i|
    \end{multline}
    for each $i\in\Nc$ as $b_i=\ell_i=\mathbf{0}_N$ and thus $\mathcal{R}_i(\ell_i,b_i)=\emptyset$ in \eqref{eq:Wtilde} and \eqref{eq:barsi}. By using the fact that, according to \eqref{eq:barsi},
    \begin{equation}
        s_i(Z_i\mathbf{1}_N,\mathbf{0}_N,y_i,e^\text{out}_i,y_i)=-e^\text{out}_i
    \end{equation}
    we find that \eqref{eq:WleqW3_a} is equal to
    \begin{equation}
        \max\Big\{0,\lambda_i\max_{\mathcal{R}\subset\mathcal{R}_i(Z_i\mathbf{1}_N,\mathbf{0}_N)}\Big|e^\text{out}_i-\sum_{l\in\mathcal{R}}Y_le^\text{out}_i\Big|\Big\}\leq\lambda_i|e^\text{out}_i|
    \end{equation}
    and thus \eqref{eq:WleqW3} holds.
\end{proof}

\begin{lemma}\label{lemma:Wflow}
	Consider the function $\tilde{W}_{i}$ as in \eqref{eq:Wtilde} and the function $H_i$ as in \eqref{eq:bargibound}. Then for all $\ell_i\in\{0,1\}^N$, $b_i\in\{0,1\}^N$, $r_i\in\R^{n_{y,i}}$, $x\in\R^{n_{x}}$, $v\in\R^{n_v}$ and almost all $e^\text{out}_i\in\R^{Nn_{y,i}}$, it holds that
	\begin{multline}
	\left\langle \frac{\partial\tilde{W}_{i}\left(\ell_i,b_i,y_i,e^\text{out}_i,r_{i}\right)}{\partial (e^\text{out}_i,y_i)},(g_i(x,e,w),f_{y,i}(x,e,w))\right\rangle \\ \leq \sqrt{N_i}H_i(x,e^\text{in}_i,v)+
	\tilde{L}_i(p(\ell_i,b_i))\tilde{W}_{i}\left(\ell_i,b_i,y_i,e^\text{out}_i,r_{i}\right),\label{eq:bounddotW-1}
	\end{multline}
	where $g_i(x,e,w):=(g^1_i(x,e,w),g^2_i(x,e,w),\ldots,\allowbreak g^N_i(x,e,w))$ with $g^m_i(x,e,w)$, $i\in\Nc$, $m\in\mathcal{V}^\text{out}_i$ as in \eqref{eq:gij}. Recall that $g^m_i:=-\delta_i(m)f_{y,i}(x,e,v)$.
\end{lemma}
\begin{proof}
    We consider the following two cases.

    \underline{Case 1:} $\tilde{W}_i(\ell_i,b_i,y_i,e^\text{out}_i,r_i)=|e^\text{out}_i+s_i(\ell_i,b_i,y_i,e^\text{out}_i,r_i)|$. For this case we have that
    \begin{multline}\label{eq:gradWcase1}
        \left\langle\frac{\partial\tilde{W}_i(\ell_i,b_i,y_i,e^\text{out}_i,r_i)}{\partial(e^\text{out}_i,y_i)},(g_i,f_{y,i})\right\rangle\\
        =\left\langle\frac{\partial|e^\text{out}_i+s_i(\ell_i,b_i,y_i,e^\text{out}_i,r_i)|}{\partial(e^\text{out}_i,y_i)},(g_i,f_{y,i})\right\rangle\\
        \stackrel{\eqref{eq:barsi}}{\leq}\sqrt{\sum_{m\in\mathcal{V}_i^\text{out}\setminus\mathcal{R}_i(\ell_i,b_i)} |g_i^m|^2+\sum_{m\in\mathcal{R}_i(\ell_i,b_i)}|f_{y,i}|^2}\\
        \hspace{-21mm}\stackrel{\eqref{eq:gij},\eqref{eq:bargibound}}{\leq}\sqrt{N_i}\left(H_i(x,e^\text{in}_i,v)+L_{i}|e_i^i|\right)\\
        \leq\sqrt{N_i}H_i(x,e^\text{in}_i,v)+L_i\sqrt{N_i}\frac{\lambda_i}{\lambda_i^{-{p(\ell_i,b_i)}}}|e^\text{out}_i|\\
        \stackrel{\eqref{eq:gamma01}}{\leq}\sqrt{N_i}H_i(x,e^\text{in}_i,v)+\tilde{L}_i(p(\ell_i,b_i))\tilde{W}_i(\ell_i,b_i,y_i,e^\text{out}_i,r_i)
    \end{multline}
    where we used the facts in the last inequality that $|e^\text{out}_i|=\tilde{W}_i(\mathbf{0}_N,\mathbf{0}_N,y_i,e^\text{out}_i,r_i)$ and thus $\lambda_i|e^\text{out}_i|\leq\lambda_i\underset{\mathcal{R}\subset\mathcal{R}_i(\ell_i,b_i)}{\max}\left\{|e^\text{out}_i+\sum_{l\in\mathcal{R}}Y_ls_i(\ell_i,b_i,y_i,e^\text{out}_i,r_i)|\right\}$. Moreover, recall that $N_i=|\mathcal{V}^\text{out}_i|$.

    \underline{Case 2:} $\tilde{W}_i(\ell_i,b_i,y_i,e^\text{out}_i,r_i)\\=\lambda_i\underset{\mathcal{R}\subset\mathcal{R}_i(\ell_i,b_i)}{\max}\Big\vert e^\text{out}_i+\sum_{l\in\mathcal{R}}Y_ls_i(\ell_i,b_i,y_i,e^\text{out}_i,r_{i})\Big\vert$ (and thus $p(\ell_i,b_i)=1$ (otherwise $\mathcal{R}_i(\ell_i,b_i)=\emptyset$ and thus $s_{i}(\ell_i,b_i,y_i,e^\text{out}_i,r_i)=0$). For this case, we define the set $\mathcal{R}^*$ as
    \begin{equation}
        \mathcal{R}^*:=\argmax_{\mathcal{R}\subset\mathcal{R}_i(\ell_i,b_i)}\Big\{\Big|e^\text{out}_i+\sum_{l\in\mathcal{R}}Y_ls_i(\ell_i,b_i,y_i,e^\text{out}_i,r_{i})\Big|\Big\}
    \end{equation}
    such that
    \begin{multline}
        \max_{\mathcal{R}\subset\mathcal{R}_i(\ell_i,b_i)}\Big\{\Big|e^\text{out}_i+\sum_{l\in\mathcal{R}}Y_ls_i(\ell_i,b_i,y_i,e^\text{out}_i,r_{i})\Big|\Big\}=\\
        \Big|e^\text{out}_i+\sum_{l\in\mathcal{R}^*}Y_ls_i(\ell_i,b_i,y_i,e^\text{out}_i,r_{i})\Big|.
    \end{multline}

    Using the definition above, we have that
    \begin{multline}\label{eq:gradWcase2}
        \left\langle\frac{\partial\tilde{W}_{i}\left(\ell_i,b_i,y_i,e^\text{out}_i,r_{i}\right)}{\partial (e^\text{out}_i,y_i)},\left(g_i,f_{y,i})\right)\right\rangle \\
        =\lambda_i\left\langle\frac{\left|e^\text{out}_i+\sum_{l\in\mathcal{R}^*}Y_ls_i(\ell_i,b_i,y_i,e^\text{out}_i,r_{i})\right|}{\partial (e^\text{out}_i,y_i)},\left({g}_i,f_{y,i}\right)\right\rangle\\
        \hspace{-26mm}\stackrel{\eqref{eq:barsi}}{\leqslant}\lambda_i\sqrt{\sum_{m\in\mathcal{V}_i^\text{out}\setminus\mathcal{R}^*}\vert g^m_i\vert^2+\sum_{m\in\mathcal{R}^*}\vert f_{y,i}\vert^2}\\
        \hspace{-35mm}\stackrel{\eqref{eq:gij},\eqref{eq:bargibound}}{\leq}\sqrt{N_i}\left(H_i(x,e^\text{in}_i,v)+L_{i}|e_i^i|\right)\\
        \hspace{-18mm}\leq\sqrt{N_i}H_i(x,e^\text{in}_i,v)+L_{i}\sqrt{N_i}\frac{\lambda_i}{\lambda_i^{-{p(\ell_i,b_i)}}}|e^\text{out}_i|\\
        \stackrel{\eqref{eq:gamma01}}{\leq}\sqrt{N_i}H_i(x,e^\text{in}_i,v)+\tilde{L}_i(p(\ell_i,b_i))\tilde{W}_i(\ell_i,b_i,y_i,e^\text{out}_i,r_i)
    \end{multline}
    where we used the fact that $\lambda_i<1$ together with the same arguments as before.

    Based on \eqref{eq:gradWcase1} and \eqref{eq:gradWcase2}, we can conclude that \eqref{eq:bounddotW-1} is true, which completes the proof of Lemma \ref{lemma:Wflow}.
\end{proof}

\begin{lemma}\label{lemma:Vflow}
	Consider the system $\mathcal{H}(\mathcal{C},F,\mathcal{D},G)$ with data $\mathcal{C}$, $F$, $\mathcal{D}$ and $G$ as described in \eqref{eq:F(xi)}-\eqref{eq:localjumpset2}, the function $V$ satisfying \eqref{eq:Vflowcondition} and the function $H_i$ as in \eqref{eq:bargibound}. Then for all $e\in\R^{Nn_{y}}$, $r\in\R^{n_y}$, $v\in\R^{n_v}$, $\ell\in\{0,1\}^{N^2}$, $b\in\{0,1\}^{N^2}$ and all $x\in\R^{n_x}$, it holds that
	\begin{multline}\label{eq:Vflow}
    	\hspace{-4mm}\left\langle \nabla V(x),f(x,e,v)\right\rangle\leq s(x,e,v)+\sum_{i\in\Nc}\Big(-\varsigma_{i}(\yhss{\text{in}}{i})
    	\\-\mu_iN_iH_i^{2}(x,e^\text{in}_i,v)+\tilde{\gamma}_i^2(p(\ell_{i},b_i))\tilde{W}^2_{i}\left(\ell_i,b_i,y_i,e^\text{out}_i,r_{i}\right)\Big).
	\end{multline}
\end{lemma}
\begin{proof}
    To prove Lemma \ref{lemma:Vflow}, based on \eqref{eq:Vflowcondition}, we need to show that
    \begin{equation}
        \tilde{\gamma}^2_i(p(\ell_i,b_i))\tilde{W}_i^2(\ell_i,b_i,y_i,e^\text{out}_i,r_i)\geq\gamma^2_{i}|e^\text{out}_i|^2.
    \end{equation}
    Recalling \eqref{eq:gamma01}, we obtain for $p(\ell_i,b_i)=0$ (and thus $\ell_i=b_i=\mathbf{0}_N$) that
    \begin{equation}
        \tilde{\gamma}_i^2(0)\tilde{W}_i^2(\mathbf{0}_N,\mathbf{0}_N,y_i,e^\text{out}_i,r_i)\stackrel{\eqref{eq:gamma01},\eqref{eq:Wtilde}}{=}\gamma_{i}^2|e^\text{out}_i|^2
    \end{equation}
    and for $p(\ell_i,b_i)=1$ that
    \begin{multline}
        \tilde{\gamma}_i^2(1)\tilde{W}_i(\ell_i,b_i,y_i,e^\text{out}_i,r_i)\stackrel{\eqref{eq:gamma01}}{=}\lambda^{-2}\gamma_{i}^2\tilde{W}_i^2(\ell_i,b_i,y_i,e^\text{out}_i,r_i)\\
        \stackrel{\eqref{eq:Wtilde}}{\geq}\gamma^2_{i}|e^\text{out}_i|^2
    \end{multline}
    for all $e^\text{out}_i\in\R^{Nn_{y,i}}$ and $r_i\in\R^{n_{y,i}}$.
\end{proof}

\textbf{Step III.} \emph{Properties of $\eta$}

As described in \eqref{eq:ETCTriggerDesign}, the dynamics of $\eta$ are governed by the functions $\Psi_i$, $\varrho_i$ and $\nu_i$ which are given in \eqref{eq:Psi2}, \eqref{eq:etai0} and \eqref{eq:etai1}, respectively. These functions are specifically designed such that the following lemma holds.
\begin{lemma}\label{lemma:etajump}
    For all $y_i\in\R^{n_{y,i}}$, $\yhss{\text{out}}{i}\in\R^{Nn_{y,i}}$ and all $\tau_i\geq\tau_\text{MIET}^i$, $i\in\Nc$, it holds that
    \begin{multline}\label{eq:etai0bound}
        \varrho_i(y_i,\yhss{\text{out}}{i})\leq-\Big(\tilde{\gamma_{i}}(1)\bar{\phi}_{1,i}(0,0)\tilde{W}^2_{i}(Z_i\mathbf{1}_N,\mathbf{0}_N,y_i,e^\text{out}_i,y_i)\\
        -\tilde{\gamma_{i}}(0)\bar{\phi}_{0,i}(\tau_i,\sigma_i)\tilde{W}^2_{i}(\mathbf{0}_N,\mathbf{0}_N,y_i,e^\text{out}_i,r_{i})\Big).
    \end{multline}
    For all $y_i\in\R^{n_{y,i}}$, $\yhss{\text{out}}{i}\in\R^{Nn_{y,i}}$, $i\in\Nc$, it holds that if $\tau_i\leq\tau_\text{MIET}^i$, $\nu_i=0$. Otherwise, if $\tau_i>\tau_\text{MIET}^i$, it holds that
    \begin{multline}\label{eq:etai1bound}
        \nu_i(y_i,\yhss{\text{out}}{i},\tau_i)\leq-\Big(\tilde{\gamma_{i}}(0)\bar{\phi}_{0,i}(\tau,0)\tilde{W}^2_{i}(\ell_i,b_i,y_i,e^\text{out}_i,r_i)\\
    	-\tilde{\gamma_{i}}(0)\bar{\phi}_{0,i}(\tau,\sigma_i)\tilde{W}^2_{i}(\ell_i,b_i,y_i,e^\text{out}_i,r_i)\Big).
    \end{multline}
\end{lemma}
\begin{proof}
    For \eqref{eq:etai0bound}, it holds that for all $e^\text{out}_i\in\R^{Nn_{y,i}}$ and all $\tau_i>\tau_\text{MIET}^i$
    \begin{multline}
        \hspace{-3mm}\varrho_i(e^\text{out}_i)=-\left(\tilde{\gamma}_{i}(1)\phi_{1,i}(0)\lambda_i^2-\tilde{\gamma}_{i}(0)\phi_{0,i}(\tau_\text{MIET}^i+\sigma_i)\right)|e^\text{out}_i|^2\\
        \hspace{8mm}\stackrel{\eqref{eq:barphi}}{=}-\left(\tilde{\gamma}_{i}(1)\bar{\phi}_{1,i}(0,0)\lambda_i^2-\tilde{\gamma}_{i}(0)\bar{\phi}_{0,i}(\tau_i,\sigma_i)\right)|e^\text{out}_i|^2\\
        \hspace{1mm}\stackrel{\eqref{eq:Wtilde},\eqref{eq:WleqW3}}{\leq}\kern-8pt-\Big(\tilde{\gamma_{i}}(1)\bar{\phi}_{1,i}(0,0)\tilde{W}^2_{i}(Z_i\mathbf{1}_N,\mathbf{0}_N,y_i,e^\text{out}_i,y_i)\\
        -\tilde{\gamma_{i}}(0)\bar{\phi}_{0,i}(\tau_i,\sigma_i)\tilde{W}^2_{i}(\mathbf{0}_N,\mathbf{0}_N,y_i,e^\text{out}_i,r_{i})\Big)
    \end{multline}
    For \eqref{eq:etai1bound}, observe that, for all $e^\text{out}_i\in\R^{Nn_{y,i}}$, if $\tau_i\leq\tau_\text{MIET}^i$, it holds that $\nu_i(e^\text{out}_i,\tau_i)=0$ due to the map \eqref{eq:Omegai}. Moreover, if $\tau_i>\tau_\text{MIET}^i$, for all $e^\text{out}_i\in\R^{Nn_{y,i}}$ it holds that
    \begin{multline}
        \nu_i(e^\text{out}_i,\tau_i)=-\tilde{\gamma}_i(0)\Big(\phi_{0,i}(\tau_\text{MIET}^i)-\phi_{0,i}(\tau_\text{MIET}^i+\sigma_i)\Big)\vert e^\text{out}_i\vert^2\\
        \hspace{3mm}\stackrel{\eqref{eq:barphi}}{=}-\tilde{\gamma}_i(0)\Big(\bar{\phi}_{0,i}(\tau_\text{MIET}^i,0)-\bar{\phi}_{0,i}(\tau_\text{MIET}^i,\sigma_i)\Big)\vert e^\text{out}_i\vert^2\\
        \hspace{-3mm}\stackrel{\eqref{eq:Wtilde},\eqref{eq:WleqW2}}{\leq}-\Big(\tilde{\gamma_{i}}(0)\bar{\phi}_{0,i}(\tau_\text{MIET}^i,0)\tilde{W}^2_{i}(\ell_i,b_i,y_i,e^\text{out}_i,r_i)\\
    	\hspace{6mm}-\tilde{\gamma_{i}}(0)\bar{\phi}_{0,i}(\tau_\text{MIET}^i,\sigma_i)\tilde{W}^2_{i}(\ell_i,b_i,y_i,e^\text{out}_i,r_i)\Big)\\
        \hspace{-8mm}=-\Big(\tilde{\gamma_{i}}(0)\bar{\phi}_{0,i}(\tau_i,0)\tilde{W}^2_{i}(\ell_i,b_i,y_i,e^\text{out}_i,r_i)\\
    	-\tilde{\gamma_{i}}(0)\bar{\phi}_{0,i}(\tau_i,\sigma_i)\tilde{W}^2_{i}(\ell_i,b_i,y_i,e^\text{out}_i,r_i)\Big)
    \end{multline}
    which completes the proof of Lemma \ref{lemma:etajump}.
\end{proof}

\textbf{Step IV.} \emph{Validate conditions of the storage function}

In this step, we verify that the function $U$ as given in \eqref{eq:LyapFunction2} is indeed a valid storage function for the supply rate $s(x,e,v)$ as described in Definition \ref{defn:Vdiss}.

\underline{Flow Dynamics of $U(\xi)$:} By combining \eqref{eq:ETCTriggerDesign}, \eqref{eq:diffeqphi01}, Lemma \ref{lemma:Wflow} and Lemma \ref{lemma:Vflow}, we obtain that for almost all $(\xi,v)\in{\mathbb{X}}\times\R^{n_v}$
\begin{multline}
    \left\langle \nabla U(\xi),F(\xi,v)\right\rangle\leq s(x,e,v)+\sum_{i\in\Nc}\Big[-\varsigma_{i}-\mu_iN_iH_i^{2}\\
    +\tilde{\gamma}^2_i(p(\ell_i,b_i))\tilde{W}_{i}^{2}\\
    +2\tilde{\gamma}_i(p(\ell_i,b_i))\bar{\phi}_{p(\ell_i,b_i),i}\tilde{W}_{i}(\sqrt{N_i}H_i+\tilde{L}_i(p(\ell_i,b_i))\tilde{W}_{i})\\
    -\tilde{\gamma}_i(p(\ell_i,b_i))\tilde{W}_{i}^{2}\times\\
    \left(2\tilde{L}_i(p(\ell_i,b_i))\bar{\phi}_{p(\ell_i,b_i),i}+ \tilde{\gamma}_i(p(\ell_i,b_i))(\bar{\phi}_{p(\ell_i,b_i),i}^2+1)\right)\\
    +\varsigma_i+(1-\epsilon_i)\mu_iN_i\underline{H}_i^2(\yhs{i})-\varphi_i(\eta_i)\Big],
\end{multline}
where we have omitted the arguments of $\tilde{W}_{i}(\ell_i,b_i,y_i,e^\text{out}_i,r_{i})$, $H_i(x,e_i^\text{in},v)$ and $\varsigma_{i}(\yhss{\text{in}}{i})$. By using the fact that for some $\epsilon_i>0$
\begin{multline*}
    2\tilde{\gamma}_i(p(\ell_i,b_i))\bar{\phi}_{p(\ell_i,b_i),i}\tilde{W}_{i}\sqrt{N_i}H_i\\
    \leq\frac{1}{\mu_i\epsilon_i}\tilde{\gamma}_i^2(p(\ell_i,b_i))\bar{\phi}_{p(\ell_i,b_i),i}^2\tilde{W}^2_{i}+\mu_i\epsilon_iN_i\underline{H}_i^2,
\end{multline*}
with $\underline{H}_i$ as in Condition \ref{cond:H}, and by substituting \eqref{eq:supplyrate} and \eqref{eq:Psi2}, we obtain
\begin{equation}
\left\langle \nabla U(\xi),F(\xi,w)\right\rangle  \leq \tilde{s}(\xi,w).
\end{equation}
Hence, $U$ satisfies Definition \ref{defn:Vdiss}.

\underline{Jump Dynamics of $U(\xi)$}: For the jump dynamics, we need to consider the following three cases.
\begin{itemize}
	\item Case 1: when $\xi\in \mathcal{D}_{i}\land\sum_{m\in\mathcal{V}^\text{out}_i}\ell^m_i=0\land\tau_i\geq\tau_\text{MIET}^i$ for some $i\in\Nc$ and $\xi$ jumps according to $\xi^+=G^{a}_i$. In this case, for a subset of agents $m\in\mathcal{M}\subset\mathcal{V}^\text{in}_i$, it may hold that $\tilde{W}_m$ is updated according to \eqref{eq:WleqW2}. Observe that for these updates, $U$ is non-increasing. Additionally, for agent $i$, we have that
	\begin{multline}
    	U(\xi^+)-U(\xi)=\\
    	\tilde{\gamma_{i}}(1)\bar{\phi}_{1,i}(0,0)\tilde{W}^2_{i}(Z_i\mathbf{1}_N,\mathbf{0}_N,y_i,e^\text{out}_i,y_i)\\
    	-\tilde{\gamma_{i}}(0)\bar{\phi}_{0,i}(\tau_i,\sigma_i)\tilde{W}^2_{i}(\mathbf{0}_N,\mathbf{0}_N,y_i,e^\text{out}_i,r_{i})+\varrho_i(e^\text{out}_i).
	\end{multline}
	Using \eqref{eq:etai0bound}, we obtain that $U(\xi^+)-U(\xi)\leq0$ for all $\xi\in\mathcal{D}_{i}$ with $\sum_{m\in\mathcal{V}^\text{out}_i}\ell^m_i=0$, for some $i\in\Nc$.
	\item Case 2: when $\xi\in \mathcal{D}_{i}\land\sum_{m\in\mathcal{V}^\text{out}_i}\ell^m_i=0$ for some $i\in\Nc$ and $\xi$ jumps according to $\xi^+=G^{b}_i$. In this case, for a subset of agents $m\in\mathcal{M}\subset\mathcal{V}_i^\text{in}$, it may hold that $\tilde{W}_m$ is updated according to \eqref{eq:WleqW2}. Observe that for these updates, $U$ is non-increasing. Additionally, for agent $i$, the following subcases hold:
    \begin{enumerate}
        \item $\tau_i\leq\tau_\text{MIET}^i$. Hence we have
        \begin{multline}
            U(\xi^+)-U(\xi)=\hspace{10mm}\\
            \tilde{\gamma}_i(p(\ell_i,b_i))\bar{\phi}_{p(\ell_i,b_i)}(\tau_i,0)\tilde{W}^2_{i}(\ell_i,b_i,y_i,e^\text{out}_i,r_{i})\\
            -\tilde{\gamma}_i(p(\ell_i,b_i))\bar{\phi}_{p(\ell_i,b_i)}(\tau_i,\sigma_i)\tilde{W}^2_{i}(\ell_i,b_i,y_i,e^\text{out}_i,r_{i}).
        \end{multline}
        Since $\bar{\phi}_{p(\ell_i,b_i)}(\tau_i,0)=\bar{\phi}_{p(\ell_i,b_i)}(\tau_i,\sigma_i)$ due to $\tau_i\leq\tau_\text{MIET}^i$, $U(\xi^+)-U(\xi)=0$ in this case.
        \item $\tau_i>\tau_\text{MIET}^i$: Hence we have
    	\begin{multline}
        	U(\xi^+)-U(\xi)=\\
        	\tilde{\gamma_{i}}(0)\bar{\phi}_{0,i}(\tau_\text{MIET}^i,0)\tilde{W}^2_{i}(\ell_i,b_i,y_i,e^\text{out}_i,r_i)\\
        	-\tilde{\gamma_{i}}(0)\bar{\phi}_{0,i}(\tau_i,\sigma_i)\tilde{W}^2_{i}(\ell_i,b_i,y_i,e^\text{out}_i,r_i)+\nu_i(e^\text{out}_i,\tau_i).
    	\end{multline}
    	Using \eqref{eq:etai1bound}, we obtain that $U(\xi^+)-U(\xi)\leq0$ for all $\xi\in\mathcal{D}_{i}\land\sum_{m\in\mathcal{V}^\text{out}_i}\ell^m_i=0$, for some $i\in\Nc$.

    \end{enumerate}
	\item Case 3: when $\xi\in \mathcal{D}_{i}\land\sum_{m\in\mathcal{V}^\text{out}_i}\ell^m_i\geq 1$ for some $i\in\Nc$ and $m\in\mathcal{V}^\text{out}_i$, and $\xi$ jumps according to $\xi^+=G^{c}_{i,m}$, $i\in\Nc$, $m\in\mathcal{V}^\text{out}_i$,
	\begin{multline}
    	U(\xi^+)-U(\xi)=\\
    	\tilde{\gamma_{i}}(p(\ell_i-\Gamma_m\mathbf{1}_N,b_i+\Gamma_m\mathbf{1}_N))\bar{\phi}_{p(\ell_i-\Gamma_m\mathbf{1}_N,b_i+\Gamma_m\mathbf{1}_N),i}(\tau_i,\sigma_i)\\
    	\times\tilde{W}^2_{i}(\ell_{i}-\Gamma_m\mathbf{1}_N,b_i+\Gamma_m\mathbf{1}_N,y_i,e^\text{out}_i,r_{i})\\
    	-\tilde{\gamma_{i}}(p(\ell_i,b_i))\bar{\phi}_{p(\ell_i,b_i),i}(\tau_i)\tilde{W}^2_{i}(\ell_i,b_i,y_i,e^\text{out}_i,r_{i}).
	\end{multline}
	Based on the fact that $p(\ell_i-\Gamma_m\mathbf{1}_N,b_i+\Gamma_m\mathbf{1}_N)=p(\ell_i,b_i)$ and \eqref{eq:WleqW1}, we can conclude that $U(\xi^+)-U(\xi)=0$ for all $\xi\in\mathcal{D}_{i}\land\sum_{m\in\mathcal{V}^\text{out}_i}\ell^m_i>1$, for some $i\in\Nc$.
\end{itemize}

\underline{Persistently flowing property}:

To verify the persistently flowing property, we first consider similar conditions as provided in \cite[Proposition 6.10]{Goebel_Sanfelice_Teel_2012} to show that each maximal solution is complete. First, we show that for any $\xi\in \mathcal{C}\setminus \mathcal{D}$ there exists a neighborhood $S$ of $\xi$ such that, it holds for every $\xi\in S\cap \mathcal{C}$ that $F(\xi,v)\cap T_{\mathcal{C}}(\varphi) \neq \emptyset$, where $T_{\mathcal{C}}(\xi)$ is the tangent cone\footnote{The tangent cone to a set $S\subset \R^{n}$ at a point $x\in\R^{n}$, denoted $T_{S}(x)$, is the set of all vectors $\omega\in\R^{n}$ for which there exist $x_{i}\in S, \tau_{i}>0$ with $x_{i}\to x, \tau \to 0$ as $i \to \infty$ such that $\omega = \lim_{i\to \infty} (x_{i} - x)/\tau_{i}$ (see Definition 5.12 in \cite{Goebel_Sanfelice_Teel_2012}).} to $\mathcal{C}$ at $\xi$.
Observe that for each $\xi\in\mathcal{C}$ for which $\ell^m_i=0$ for all $i,m\in\Nc$ (recall that $\xi=\left(x,e,\tau,\sigma,r,\ell,b,\eta\right)$), $T_{\mathcal{C}}(\xi) = \R^{n_x}\!\times\!\R^{Nn_y}\!\times\!(T_{\R_{\geq 0}}(\tau_1)\times\ldots\times T_{\R_{\geq 0}}(\tau_N))\!\times\!(T_{\R_{\geq 0}}(\sigma_1)\times\ldots\times T_{\R_{\geq 0}}(\sigma_N))\!\times\!\R^{Nn_y}\!\times\!\{0\}^{N^2}\!\times\!\{0\}^{N^2}\!\times\!(T_{\R_{\geq 0}}(\eta_1)\!\times\!\ldots\!\times\! T_{\R_{\geq 0}}(\eta_N))$. Observe also from \eqref{eq:localflowset} and \eqref{eq:localjumpset} that $\mathcal{C}\setminus\mathcal{D}=\bigcap_{i,m\in\Nc}\{\xi\in\mathbb{X}:\ell^m_i=0\land(\sigma_i<d_i\lor\eta_i>0)\}$.
Given the facts that, according to \eqref{eq:Flow} and \eqref{eq:F(xi)}, for all $i\in\Nc$, $\dot{\tau}_i = 1$, $\dot{\sigma}_i=1$ and that $\dot{\eta_i}\geq0$ when  $\eta_i=0$ due to \eqref{eq:ETCTriggerDesign}, it indeed follows that for any $\xi\in \mathcal{C}\setminus\mathcal{D}$ there exists a neighborhood $S$ of $\xi$ such that, it holds for every $\varphi\in S\cap \mathcal{C}$ that $F(\varphi,w)\cap T_{\mathcal{C}}(\varphi) \neq \emptyset$.

Due to the hypothesis that there are no finite escape times during flow, case (b) in \cite[Prop.~6.10]{Goebel_Sanfelice_Teel_2012} is ruled out. Lastly, observe from \eqref{eq:localjumpset} that $G(\mathcal{D})\subset\mathcal{C}\cup\mathcal{D}$ since for all $\xi\in G(\mathcal{D})$, it holds that $\tau_{i}^+\geq0$, $\eta_{i}^+\geq0$ since $\varsigma_i(\yhss{\text{in}}{i})\geq0$ for all $\yhss{\text{in}}{i}\in\R^{n_y}$. As a consequence case (c) in \cite[Prop.~6.10]{Goebel_Sanfelice_Teel_2012} cannot occur and all maximal solutions are complete. In addition, since $d_i>0$ and the number of agents that are receiving and transmitting information is finite, it can be shown that solutions have a (global) average dwell-time; the details are omitted. Therefore, the system is persistently flowing, which completes the proof of Theorem \ref{thm:Theorem}.
 \qed
\end{document}